\documentclass[lettersize,journal]{IEEEtran}
\usepackage{amsmath,amsfonts}
\usepackage{algorithmic}
\usepackage{algorithm}
\usepackage{array}
\usepackage[caption=false,font=normalsize,labelfont=sf,textfont=sf]{subfig}
\usepackage{textcomp}
\usepackage{stfloats}
\usepackage{hyperref}
\usepackage{url}
\usepackage{verbatim}
\usepackage{graphicx}
\usepackage{cite}
\usepackage{multirow}
\usepackage{multicol}
\usepackage{hyperref}
\usepackage{doi}

\usepackage{amssymb}
\usepackage{mathtools}
\usepackage{amsthm}
\usepackage{pifont}
\usepackage[T1]{fontenc}
\usepackage{cleveref}
\usepackage{subcaption}
\usepackage{multirow}
\usepackage{xcolor}
\usepackage{diagbox}
\usepackage{natbib}
\setcitestyle{numbers,square}

\theoremstyle{plain}
\newtheorem{theorem}{Theorem}[section]

\newtheorem{lemma}[theorem]{Lemma}
\newtheorem{corollary}[theorem]{Corollary}
\theoremstyle{definition}

\newtheorem{assumption}[theorem]{Assumption}
\theoremstyle{remark}
\newtheorem{remark}[theorem]{Remark}

\begin{document}

\title{PE-MA: Parameter-Efficient Co-Evolution of Multi-Agent Systems}

\author{Yingfan Deng,~\IEEEmembership{Student Member,~IEEE,} Anhao Zhou, ~\IEEEmembership{Student Member,~IEEE,} Yuan Yuan,~\IEEEmembership{Member,~IEEE,}  Xiao Zhang,~\IEEEmembership{Member,~IEEE,}  Yifei Zou,~\IEEEmembership{ Member,~IEEE,} Dongxiao Yu,~\IEEEmembership{Senior Member,~IEEE}
\thanks{Y.Deng, A.Zhou, X.Zhang, Y.Zou and D.Yu are with the School of Computer Science and Technology, Shandong University, Qingdao, 266237, China.\protect ~E-mail: \{dengyf,zhouah\}@mail.sdu.edu.cn; \{xiaozhang, yfzou, dxyu\}@sdu.edu.cn}\\
\thanks{Y.Yuan is with the School of Software $\&$ Joint SDU-NTU Centre for Artificial Intelligence Research (C-FAIR), Shandong University, Jinan, 250000, China. \protect ~E-mail: yyuan@sdu.edu.cn}
}



\maketitle
\begin{abstract}
Multi-Agent Systems have recently emerged as a promising paradigm for collaborative reasoning and solving complex tasks. However, the design of collaborative learning algorithms in multi-agent systems faces several challenges, including high communication overhead and insufficient agent-level personalization. In this paper, we propose \emph{PE-MA} (Parameter-Efficient Multi-Agent Co-Evolution), a novel collaboration framework that supports efficient, scalable, and personalized co-evolution in multi-agent systems. In \emph{PE-MA}, each agent maintains a lightweight personalized adapter to support agent-specific behavior, while a shared adapter is collaboratively optimized across neighboring agents. This design balances global coordination with local adaptation under heterogeneous environments. We achieve an asymptotically optimal convergence rate of $\mathcal{O}(\frac{1}{\sqrt{NK}})$, where $N$ is the number of agents and $K$ the local update steps. Experiments show that \emph{PE-MA} improves accuracy by 2\%–5\%, while reducing training and communication costs by 77\% and 87\%, respectively.
\end{abstract}


\section{Introduction}
\IEEEPARstart{I}{n} recent years, multi-agent systems (MAS) have gradually become an important research topic in the field of artificial
intelligence due to breakthroughs in natural language understanding and generation. These systems organize multiple agents with communication capabilities into collaborative frameworks, enhancing their abilities in task decomposition, role allocation, and collective reasoning. This drives the paradigm of solving complex tasks toward greater intelligence and autonomy.  
It has been widely applied in intelligent dialogue systems \cite{park2023generative,hong2023metagpt}, automated decision-making \cite{wang2023voyager}, robotic collaboration \cite{burgard2005coordinated,huang2022inner}, and virtual assistants \cite{yin2023lamm,shen2023hugginggpt},  demonstrating significant application potential.

Despite the powerful capabilities of multi-agent system, their large scale, high training costs, and opaque reasoning paths present a series of new challenges \cite{han2024llm} in multi-agent collaborative learning, as show in Fig.~\ref{fig:challenge}. 1) \textbf{High computation and communication costs:} Due to the vast number of parameters, the cost of knowledge sharing and model updating between individual agents is extremely high, making efficient federated or distributed collaborative optimization difficult. 2) \textbf{Lack of efficient co-evolutionary learning mechanisms:} Conventional multi-agent training paradigms predominantly rely on either fully independent optimization or uniform parameter sharing across agents, both of which fall short in enabling fine-grained and adaptive collaboration. 3) \textbf{Decision consistency:} The personalization of local tasks leads to inherent differences in the information and knowledge possessed by individual agents, creating information asymmetry that hinders coordinated decision-making and degrades overall collaborative efficiency \cite{gao2024large}.  \emph{Therefore, it is an urgent problem to design efficient multi-agent co-evolutionary learning algorithms that are tailored to local personalized tasks while enabling agents to evolve through dynamic and continuous knowledge sharing with related agents.}

Although recent studies have explored collaborative learning in multi-agent systems, existing approaches suffer from critical limitations in both communication efficiency and adaptability. Early communication-based methods (e.g., CommNet~\cite{sukhbaatar2016learning}, DIAL~\cite{foerster2016learning}) rely on dense feature sharing, which is impractical due to token-level overhead and lack of semantic abstraction. Later works (e.g., TarMAC \cite{das2019tarmac}, OPV2V~\cite{xu2022opv2v}) improve the overall performance of the system through effective communication mechanism design, but the limited and biased data sets may lead to model overfitting.
To reduce communication costs, knowledge distillation frameworks such as DiscoNet~\cite{li2021learning} and MKD-Cooper~\cite{li2023mkd} have been proposed. However, these models are typically static and do not support dynamic, context-aware integration of peer knowledge. Moreover, they often compress information excessively, leading to loss of semantic richness and adaptability in evolving environments.
While parameter-efficient fine-tuning (PEFT) methods such as LoRA and Adapter Tuning~\cite{hu2021lora,houlsby2019parameter} significantly reduce adaptation cost, they are mainly designed for single-agent or centralized settings and lack support for inter-agent coordination. In summary, a unified framework is still lacking for efficient co-evolution in multi-agent systems — particularly under heterogeneous and decentralized conditions.

\begin{figure}[h]
	\vskip 0.2in
	\begin{center}
		\includegraphics[width=0.8\columnwidth]{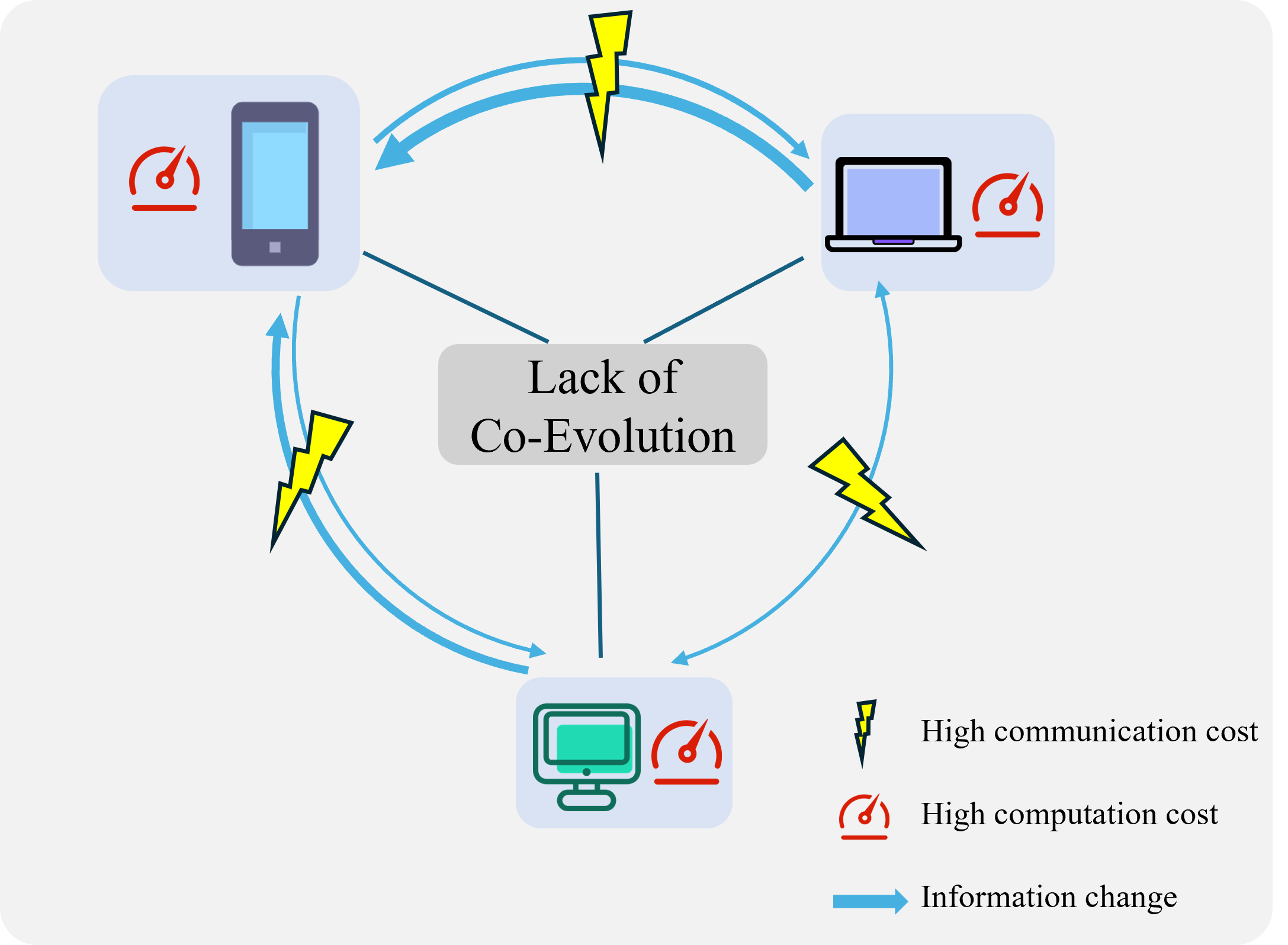}
		\caption{Challenges of existing multi-agent frameworks} 
		\label{fig:challenge}
	\end{center}
	\vskip -0.2in
\end{figure}

To address these challenges, we propose \emph{PE-MA}, a novel framework that enables parameter-efficient and personalized co-evolution in multi-agent systems. \emph{PE-MA} is specifically designed to reduce both communication and computation overhead while facilitating effective knowledge sharing and continual adaptation across heterogeneous agents. At the core of \emph{PE-MA} is a dual-adapter architecture, in which each agent maintains a lightweight personalized adapter to specialize the whole model to its private data and a shared  adapter to capture and exchange common knowledge with its agents. This design allows agents to evolve independently while simultaneously benefiting from semantically meaningful insights distilled from neighboring agents.
\emph{To ensure parameter efficiency}, each agent updates and transmits only lightweight adapters rather than full model parameters. This design significantly lowers the computational burden of fine-tuning and dramatically reduces bandwidth consumption during inter-agent communication. 
\emph{To support co-evolutionary learning}, the shared adapter is continuously updated via communication with other agents. Instead of static coordination, agents dynamically incorporate semantically rich knowledge from their peers during training. 
\emph{To design effective collaboration and knowledge sharing mechanisms}, we designed a doubly stochastic matrix $\mathbf{P}$. This ensures the system can stably converge to a consistent global state. Specifically, the doubly stochastic property of the weight matrix can ensure that agents achieve an average consensus model after sufficient communication rounds. Therefore, we summarize our contributions as follows:

\begin{itemize}
    \item We design a parameter-efficient dual-adapter framework in which each agent independently learns a lightweight personalized adapter while collaboratively optimizing a shared adapter through partial communication. 
    \item We demonstrate the effectiveness of the proposed algorithm through convergence analysis. Specifically, we prove that \emph{PE-MA} achieves an asymptotically optimal convergence rate of $\mathcal{O}(\frac{1}{\sqrt{NK}})$, where $N$ denotes the number of clients and $K$ is the global maximum number of iterations.
    \item To verify the effectiveness and communication efficiency of \emph{PE-MA}, we conducted classification task experiments on three datasets. Experimental results show that \emph{PE-MA} outperforms other multi-agent algorithms in a decentralized framework. In terms of accuracy, our best result is $2-5\%$ higher than other algorithms, and our training and communication cost are reduced by $82\pm5\%$ under the same conditions. 
\end{itemize}

\section{Related work}
\textbf{Decentralized multi-agent collaborative learning} 
In decentralized multi-agent learning, how to achieve effective cooperation between agents without relying on a central controller is a fundamental challenge. One of the key differences is whether to share raw observations early \cite{zhang2021emp,luo2023edgecooper} or to communicate perception results late \cite{miller2020cooperative}. These methods cannot trade off between communication efficiency and performance \cite{li2021prefix}. So early on, there were some works that focused on learning shared message representations from hidden features CommNet\cite{sukhbaatar2016learning} and DIAL \cite{foerster2016learning}), while later methods, such as TarMAC\cite{das2019tarmac}, introduced structured attention mechanisms and selective communication to convey more abstract and decision-related signals. Attfuse \cite{xu2022opv2v} introduced self-attention to aggregate intermediate features from different agents and release high-quality OPV2V datasets to mitigate the adverse effects of adaptive errors. Meanwhile, in order to preserve the advantages of early collaboration while reducing bandwidth, Disconet\cite{li2021learning} and MKD-Cooper \cite{li2023mkd} introduced knowledge distillation to guide the learning of intermediate collaborative models. These methods aim to strike a balance between rich information sharing and parameter efficiency.

\textbf{Parameter-Efficient Fine-Tuning} 
In recent years, Parameter-Efficient Fine-Tuning (PEFT) has emerged as a key paradigm to reduce the computational and storage cost of adapting large-scale models to downstream tasks. Instead of updating all model parameters, PEFT methods aim to modify only a small subset or add compact modules to the backbone, thereby preserving generalization while achieving high task-specific performance \cite{hu2021lora}. A notable branch of PEFT techniques is prompt-based tuning, where learnable tokens are prepended to the input sequence, allowing the model to steer its attention and predictions based on these prompts \cite{lester2021power,zhou2022learning}. Building on the idea of soft prompts, recent efforts have expanded this approach to vision-language settings by injecting prompt tokens into visual transformers \cite{jia2022visual,wang2023voyager}. Beyond input manipulation, other approaches insert lightweight modules, such as adapters, into intermediate Transformer layers to locally modulate feature representations without interfering with the backbone parameters \cite{houlsby2019parameter,zhang2022adaptive}. Further innovations target other architectural components like normalization layers attention heads \cite{mahabadi2021compacter} to enhance parameter efficiency. Despite these advancements, most PEFT strategies are developed under single-agent settings, limiting their applicability to multi-agent or distributed perception tasks. 

Based on the above research, we propose a dual-adapter architecture approach, where each agent maintains a private adapter to encode local task-specific knowledge during training, while sharing a shared adapter to capture shareable patterns. By transmitting only the shared adapter, we achieve a balance between personalization and generalization while minimizing communication overhead, facilitating efficient decentralized multi-agent learning.

\begin{figure}[t]
	\begin{center}
		\includegraphics[width=1.0\columnwidth]{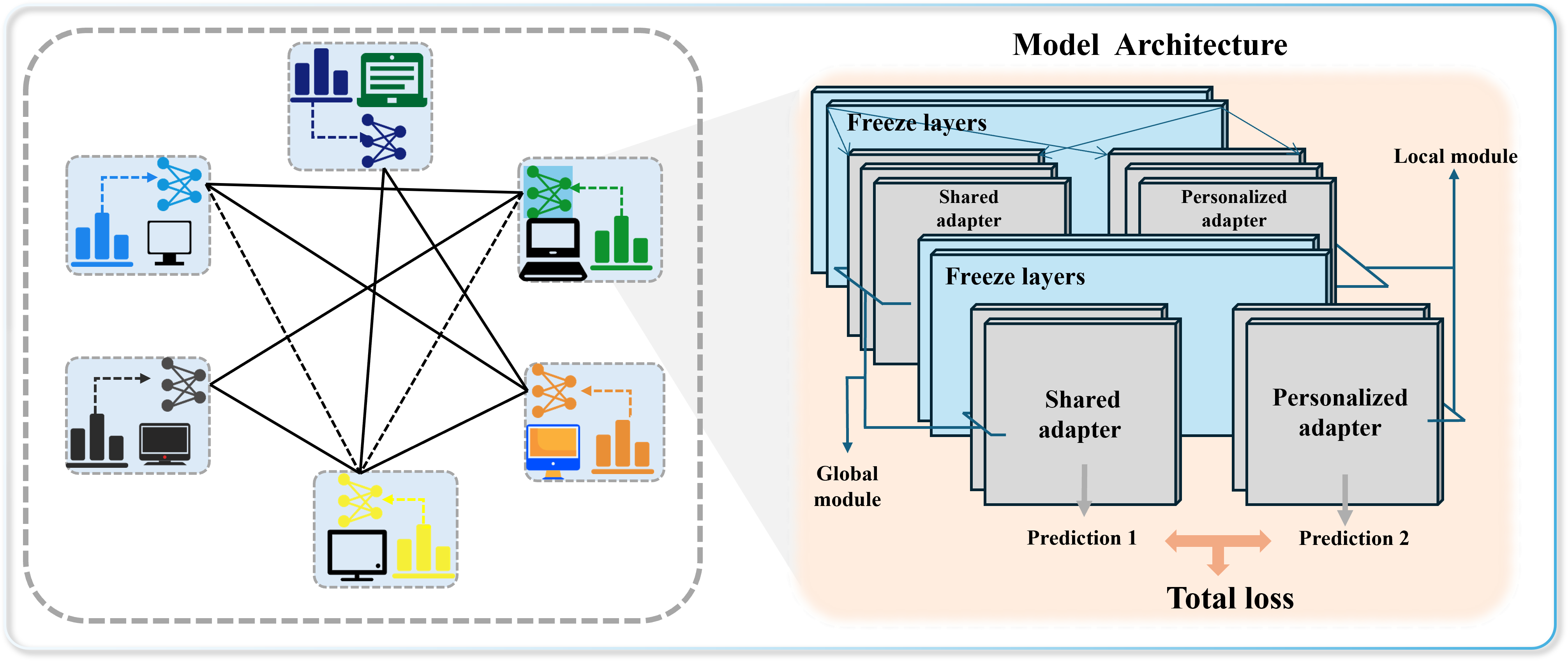}
	\caption{Multi-Agent Co-Evolution Architecture with Dual-adapter.} 
		\label{fig:overview}
	\end{center}
	\vskip -0.2in
\end{figure}

\section{Model and Problem Definition}
 In this section, we first give the model formulation, then we show the optimization objective.
\subsection{Model Formulation}
To facilitate modeling, we abstract the above system structure into a fundamental and general model framework.

In a Multi-Agent system, we consider multi-agent network composed of \( N \) agents, denoted as the set \( \mathcal{A} = \{a_1, a_2, \ldots, a_N\} \). Each agent \( a_i \) is equipped with a personalized agent model with a pre-trained module as its backbone, possessing perception, comprehension, and reasoning capabilities to handle local tasks. The system aims to enhance overall performance through efficient collaborative learning mechanisms while preserving individual adaptability.

\subsubsection{Agent Model Architecture}
Each agent \( a_i \) consists of the following key components, as show in Fig.~\ref{fig:overview}:
\begin{itemize}
    \item {Frozen Backbone $\pmb{\theta}_\text{freeze}$}: A fixed parameter general module for feature extraction of data. 
    \item {Local Personalized Adapter \( \mathbf{v}_i \)}: A small trainable module adapting to local personalized tasks and data distributions, trained locally without communication.
    \item {Shared Adapter  \( \mathbf{w}_i \)}: A multi-head module for cross-agent knowledge sharing, enabling parameter exchange and aggregation with neighboring agents.
\end{itemize}

The complete model for agent \( a_i \) is expressed as:
\[
f_i(x; \pmb{\theta}_\text{freeze}, \mathbf{w}_i, \mathbf{v}_i) 
\]
where $\pmb{\theta}_\text{freeze}$ represents the freezing parameters of the pre-trained module, $\mathbf{w}_i$ represents parameters of shared adapter, $\mathbf{v}_i$ is the parameters of personalized adapter class.

\subsubsection{Collaboration and Communication Mechanism}
Agent interactions are modeled via a graph structure \( G = (\mathcal{A}, \mathcal{E}) \), where the edge set \( \mathcal{E} \subseteq \mathcal{A} \times \mathcal{A} \) defines communication relationships. During each communication round, agent \( a_i \) transmits its shared adapter \( \mathbf{w}_i\) to neighbors and updates through weighted averaging:
\begin{align}
\mathbf{w}_i \leftarrow \sum_{j \in \mathcal{N}_i} P_{ij} \mathbf{w}_j
\end{align}
where \( \mathcal{N}_i \) is the neighbor set of agent \( i \), and \( P_{ij} \) corresponds to the doubly stochastic communication weight matrix \( \mathbf{P} \). For the weight matrix $\mathbf{P}$, $P_{ij} > 0$ if $j \in \mathcal{N}_i$ or $P_{ij}=0$. And $\mathbf{P}^\top=\mathbf{P},\mathbf{P1}=\mathbf{1}$. 

\subsection{Optimization Objective}
Each agent possesses a local dataset \( D_i \sim P_i \), where the distribution \( P_i \) may exhibit significant heterogeneity across tasks or environments. Let \( (x, y) \in D_i \), the system optimizes the balance between local personalization and global consistency via:
\begin{align}
   \min_{\mathbf{w}, \mathbf{v}_i} L(\mathbf{w},\{\mathbf{v}_i\}_{i=1}^N)=\frac{1}{N}\sum_{i=1}^N L_i(\pmb{\theta}_\text{freeze}, \mathbf{w}_i, \mathbf{v}_i) \\
   L_i(\pmb{\theta}_\text{freeze}, \mathbf{w}_i, \mathbf{v}_i)=\mathbb{E}_{(x,y) \sim D_i} \left[ \ell _i\left(f_i(x; \pmb{\theta}_\text{freeze}, \mathbf{w}_i, \mathbf{v}_i), y\right) \right]
\end{align}
where \( \ell_i \) is the task-specific loss function of agent $a_i$ (e.g., cross-entropy or mean squared error), $L_i$ is the local optimization objective of agent $a_i$ and ${L}$ is the global optimization objective.

\section{Algorithm}

This section briefly summarizes the entire algorithm. In the local training phase, the parameters of the freeze layer are fixed, and the two adapters are trained and updated at the same time to ensure that the information is retained locally and used for communication in subsequent operations. In the communication phase, we take into account both communication efficiency and co-evolution, and only the \emph{shared adapter} is communicated between agents. Through mutual promotion between agents, a decentralized collaborative learning framework \emph{PE-MA} based on a dual adapter design is implemented.

\subsection{Dual-adapter Module}
In multi-agent learning, there are two main challenges. First, each agent may have unique personalized requirements that necessitate custom modeling for local tasks. Second, it is essential to share and transfer general global knowledge among agents to enhance the model’s ability to generalize. To address these challenges, the dual-adapter module aims to strike a balance between local personalization and global knowledge sharing. This module comprises:
a \emph{self-improved module} addressing agent-specific learning needs, and a \emph{co-evolution module} facilitating cross-agent knowledge transfer, as shown in Fig.~\ref{fig:dualada}. 
This design allows each agent to retain its own personalized knowledge while learning from the cognitive knowledge of other agents, thereby facilitating co-evolution among agents and leading to enhanced overall learning performance.
\begin{figure}[h]
	\begin{center}
		\includegraphics[width=1.0\columnwidth]{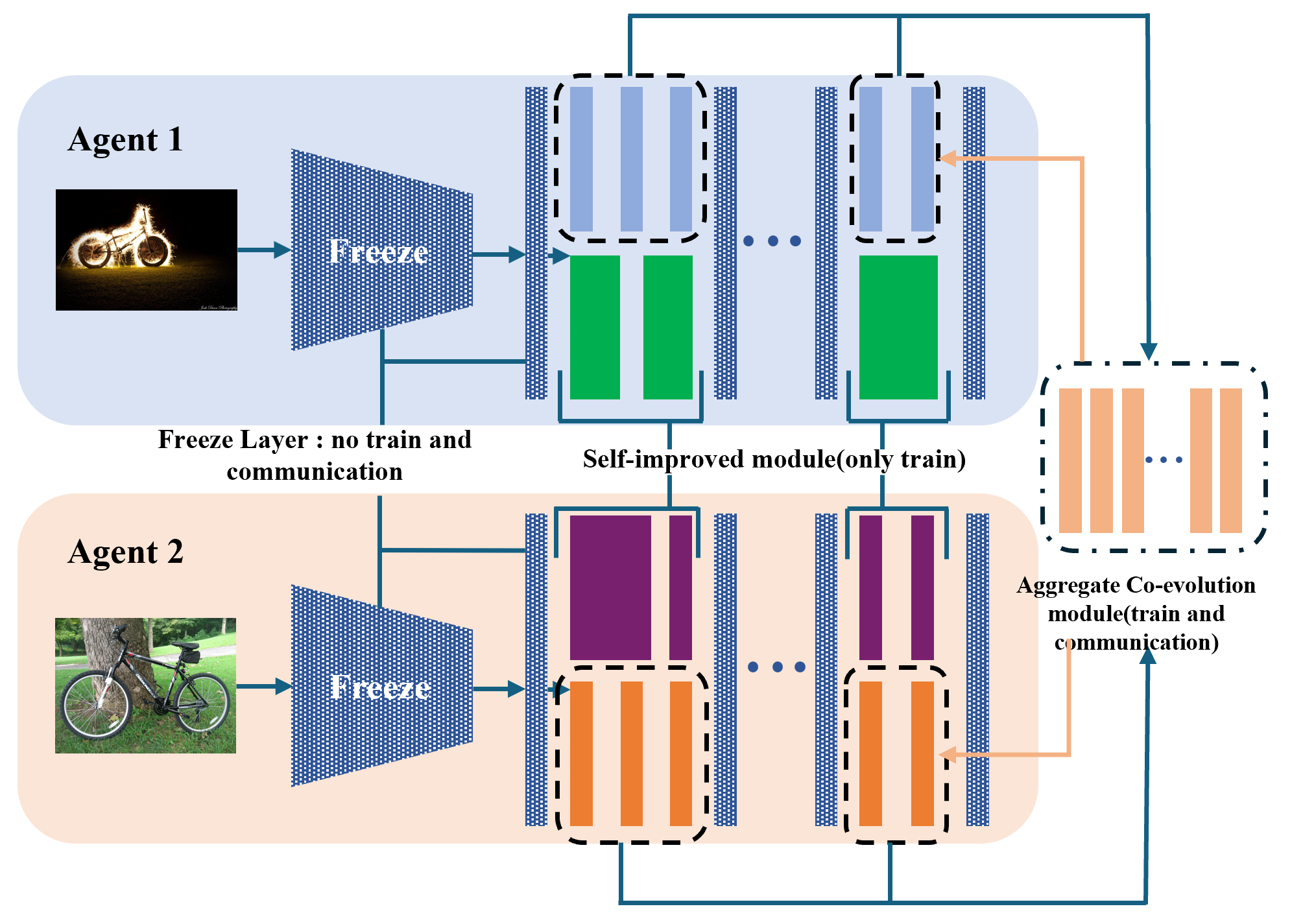}
		\caption{A brief overview of the dual-adapter module}
		\label{fig:dualada}
	\end{center}
	\vskip -0.2in
\end{figure}

\textbf{Self-improved Module:} The personalized adapter $\mathbf{v}_i $ is designed to capture agent-specific features by inserting multiple adapter modules into the backbone. These modules help each agent extract and retain distinctive patterns in its representations. By placing adapters at different layers (e.g., within residual blocks), the architecture enables multi-level control over the information flow, enhancing adaptability to diverse tasks and environments.
Its modular design allows flexible adjustment based on each agent’s needs, enabling effective personalized training. Importantly, the personalized adapter is only used locally during training and inference and is not shared across agents, ensuring strong personalization while avoiding information leakage or interference.

\textbf{Co-Evolution Module:} 
The shared adapter enables global knowledge transfer among agents. It functions as a lightweight, exchangeable module embedded in key model layers and operates in parallel with the personalized adapter. By sharing only a low-rank parameters $\mathbf{w}_i$, agents can efficiently communicate and integrate global patterns without altering their backbone networks. This design supports scalable and collaborative learning with minimal communication overhead.

\subsection{Collaborative Training with Dual-adapter}
In decentralized multi-agent systems with heterogeneous data, it is crucial to enable both local personalization and cross-agent generalization. Each agent has access only to its own private data and operates without centralized coordination, making the design of effective collaboration mechanisms essential.
To this end, we propose a collaborative training strategy based on \textbf{dual-adapter}, where each agent simultaneously optimizes a \emph{personalized adapter} for local adaptation and a \emph{shared adapter} for inter-agent knowledge sharing. The training process alternates between local model updates and communication-based aggregation, aiming to improve overall learning performance while preserving local specialization.

\paragraph{Computation of the Loss Function.}
During each communication round, agents first perform local training based on their private datasets. Instead of optimizing the local and global modules separately, we adopt a joint training strategy that coordinates their behavior via a mixing coefficient \( \mu \). Given an input sample \( (x, y) \), agent \( i \) computes two predictions using its personalized adapter \( v_i \) and shared adapter \( w_i \), respectively:
\begin{align}
    \hat{y}_1 &= f_i(x; \pmb{\theta}_{\text{freeze}}, \mathbf{w}_i), \\
    \hat{y}_2 &= f_i(x; \pmb{\theta}_{\text{freeze}}, \mathbf{v}_i),
\end{align}
where \( \pmb{\theta}_{\text{freeze}} \) denotes the frozen backbone shared across all agents. These outputs are then linearly combined:
\begin{align}
    \hat{y} = \mu \cdot \hat{y}_2 + (1 - \mu) \cdot \hat{y}_1,
\end{align}
and the overall loss is computed as:
\begin{align}
    {L}_i = \ell_i(\hat{y}, y),
\end{align}
where \( \ell_i(\cdot) \) is the standard prediction loss of agent $a_i$ (e.g., cross-entropy). The hyperparameter \( \mu \in [0,1] \) controls the relative contribution of the personalized and shared adapters. This coordinated prediction strategy enhances the expressiveness of local preferences while incorporating generalized knowledge from other agents.

\paragraph{Adapter Parameter Update.}
After loss computation, both adapter modules are updated via gradient descent:
\begin{align}
    \mathbf{w}_{i}(t+\frac{1}{2}) &= \mathbf{w}_{i}(t) - \eta_w \nabla_{\mathbf{w}} L_i, \\
    \mathbf{v}_{i}(t,q+1) &= \mathbf{v}_{i}(t,q) - \eta_v \nabla_{\mathbf{v}} L_i,
\end{align}
where \( \eta_\mathbf{w} \) and \( \eta_\mathbf{v} \) denote learning rates for the global and personalized adapters, respectively. During algorithm training, each agent trains on its private dataset for $\tau$ epochs within each communication round \( t \). We take the shared adapter as an example, where $g_{\mathbf{w}}(\mathbf{w}_{i,q},\mathbf{v}_{i,q})$ is the stochastic gradient estimate of the shared adapter:
\begin{align}
    \mathbf{w}_{i}(t+\frac{1}{2}) \leftarrow \mathbf{w}_{i}(t) - \eta_w g_{\mathbf{w}}(\mathbf{w}_{i}(t),\mathbf{v}_{i}(t,\tau))
\end{align}
\begin{align}\label{eq:local-subgradient1}
&g_{\mathbf{w}}(\mathbf{w}_i(t), \mathbf{v}_i(t,\tau))\nonumber\\
&\!\!\!\!:=\!\frac{1}{|C_i(k,s)|} \!\sum\limits_{(x_m, y_m)\in C_i(k,s)} \!\!\!\!\!\!\!\!\!\!\!\nabla_{\mathbf{w}} L_i(\mathbf{w}_i(t), \mathbf{v}_i(t,\tau), x_m, y_m),
\end{align}
where $\mathcal{C}_i(k,s)$ is a random subset of $\mathcal{D}_i$.  We allow each worker to perform $\tau$ steps local updates to find the optimal self-improved module based on its local data. This allows the model to leverage both personalized features and shared features. 

\paragraph{Adapter Parameter sharing.} Upon completion of local training, agents enter the communication phase. Each agent broadcasts its updated shared adapter \( w_i \) to neighboring agents \( \mathcal{N}(i) \), and performs weighted averaging to obtain the new shared adapter:
\begin{align}
    \mathbf{w}_i(t+1) \leftarrow \frac{1}{|\mathcal{N}(i)|} \sum_{j \in \mathcal{N}(i)} P_{ij} \cdot \mathbf{w}_j(t+\frac{1}{2}),
\end{align}
where \( P_{ij} \) denotes the weight assigned to neighbor \( j \), which can be set based on data similarity, communication topology, or uniformly. Importantly, the specific adapters \( v_i \) are not shared, preserving the personalized behavior of each agent while enabling efficient global collaboration through the shared adapter.

This alternating local update and communication process ensures that each agent benefits from shared knowledge while retaining the capacity to adapt to its unique data distribution.


\begin{algorithm}[t] 
\caption{PE-MA:Nested update version}
\label{alg:PE-MA1}
\begin{algorithmic}[1]
\STATE \textbf{Initialize:} $N$ agents; frozen pretrained model parameters $\pmb{\theta}_{\text{freeze}}$; initial shared adapter $w^0$; specific adapters $\{v_i^0\}_{i=1}^N$; local training epochs $\tau$; local datasets $\{\mathbb{D}_i\}_{i=1}^N$; mixing coefficient $\mu$; learning rates $\eta_w$, $\eta_v$
\STATE \textbf{Output:} $\{v_i^T\}_{i=1}^{N}, \{w_i^T\}_{i=1}^{N}$
\FOR{each round $t = 0, 1, \dots, K-1$}
    \STATE \textbf{// Local Training Phase}
    \FOR{each agent $i = 1, 2, \dots, N$}
        \FOR{each local epoch $q = 0, 1, \dots, \tau-1$}
            \FOR{each sample $(x, y) \in \mathbb{D}_i$}
            \STATE $\hat{y}_1 \leftarrow f_i(x; \pmb{\theta}_{\text{freeze}}, w_i(t))$ \hfill \textit{// via shared-adapter}
            \STATE $\hat{y}_2 \leftarrow f_i(x; \pmb{\theta}_{\text{freeze}}, v_i(t,q))$ \hfill \textit{// via personalized-adapter}
            \STATE $L_i(x,y) \leftarrow \ell_i\left(\mu \cdot \hat{y}_2 + (1-\mu) \cdot \hat{y}_1, y\right)$
        \ENDFOR
            \STATE Update personalized adapter: $\mathbf{v}_{i}(t,q+1) \leftarrow \mathbf{v}_{i}(t, q) - \eta_v g_{\mathbf{v}}(\mathbf{w}_{i}(t), \mathbf{v}_{i}(t,q))$
        \ENDFOR
            \STATE Update shared adapter: $\mathbf{w}_{i}(t+\frac{1}{2}) \leftarrow \mathbf{w}_{i}(t) - \eta_w g_{\mathbf{w}}(\mathbf{w}_{i}(t), \mathbf{v}_{i}(t,\tau))$
            \STATE Set $v_i(t+1,0) = v_i(t, \tau).$ 
    \ENDFOR
    \STATE \textbf{// Communication Phase}
    \FOR{each agent $i = 1, 2, \dots, N$}
        \STATE Broadcast $\mathbf{w}_{i}(t + \frac{1}{2})$ to neighbors $\mathcal{N}_i$
        \STATE Aggregate shared adapter: $\mathbf{w}_{i}(t+1) \leftarrow  \sum_{j \in \mathcal{N}_i} \mathbf{w}_{j}(t+\frac{1}{2})P_{ij}$
    \ENDFOR
\ENDFOR
\end{algorithmic}
\end{algorithm}

\begin{table*}[h]
\caption{Performance of different methods on three datasets}
\centering
\renewcommand{\arraystretch}{1.4}
\scalebox{1}{
\begin{tabular}{ccccccccccc}
\hline
\multicolumn{2}{c}{\textbf{Dataset}}&\multicolumn{3}{c}{Office-Home(150 epochs)}&\multicolumn{3}{c}{Office-Caltech10(30 epochs)}&\multicolumn{3}{c}{DomainNet(80 epochs)} \\
\multicolumn{2}{c}{Frame / Method / Topology}&FC$_{(\%)}$&ER$_{(\%)}$&Ring$_{(\%)}$&FC$_{(\%)}$&ER$_{(\%)}$&Ring$_{(\%)}$&FC$_{(\%)}$&ER$_{(\%)}$&Ring$_{(\%)}$\\
\hline
\multirow{1}{*}{Indep} &Full model & 80.53&80.53&80.53&97.09
&97.09&97.09&60.13&60.13&60.13\\
\hline
\multirow{3}{*}{DSGD\_SIM}&Input layer&81.54&81.33&77.18&97.29&97.18&97.18&60.02&60.18&60.07\\
&Out layer&81.88&81.68&81.76&96.24&96.12&96.11&62.97&61.95&62.64\\
&Adapter&79.32&79.37&79.04&97.26&97.15&97.26&63.38&63.51&63.57\\
\hline
\multirow{3}{*}{DSGD\_ALT}&Input layer&81.42&81.17&81.39&97.18&97.29&97.09&60.15&60.15&60.35\\
&Out layer&81.55&81.35&81.89&95.99&95.75&96.11&62.64&61.38&61.90\\
&Adapter&80.58&78.79&79.11&96.59&96.68&96.24&62.62&61.35&62.84\\
\hline
\multirow{1}{*}{\textbf{PE-MA}} &Dual Adapter&\textbf{82.18}&\textbf{82.32}&\textbf{82.29}&\textbf{97.46}&\textbf{97.68}&\textbf{97.79}&\textbf{64.08}& \textbf{64.33}& \textbf{64.47}\\
\hline
\end{tabular}}
\label{tab:main-results}
\end{table*}

\begin{table*}[h]
\caption{The communication and training parameter efficiency of different client-side personalization methods}
\centering
\renewcommand{\arraystretch}{1.4}
\scalebox{1}{
\begin{tabular}{ccccccccccc}
\hline
\multicolumn{2}{c}{\textbf{Dataset}}&\multicolumn{2}{c}{Office-Home}&\multicolumn{2}{c}{Office-Caltech10}&\multicolumn{2}{c}{DomainNet} \\
\multicolumn{2}{c}{Frame / Method / Topology}&Comm 
Params&Train Params&Comm Params&Train Params &Comm Params&Train Params\\
\hline
\multirow{1}{*}{Indep} &ResNet18 &0M&11.21M&0M&11.18M&0M&11.35M\\
\hline
\multirow{3}{*}{DSGD\_SIM}&Input layer&6.48M&11.21M&6.45M&11.18M&6.62M&11.35M\\
&Out layer&11.17M&11.21M&11.17M&11.18M&11.17M&11.35M\\
&Adapter & 11.21M & 12.61M& 11.17M & 12.58M& 11.16M & 12.75M\\
\hline
\multirow{3}{*}{DSGD\_ALT}&Input layer&6.48M&11.21M&6.45M&11.18M&6.62M&11.35M\\
&Out layer&11.17M&11.21M&11.17M&11.18M&11.17M&11.35M\\
&Adapter & 11.21M & 12.61M& 11.17M & 12.58M& 11.16M & 12.75M\\
\hline
\multirow{1}{*}{\textbf{PE-MA}} &Dual Adapter&\textbf{1.44M}&\textbf{2.88M}&\textbf{1.41M}&\textbf{2.82M}&\textbf{1.59M}&\textbf{3.18M}\\
\hline
\end{tabular}}
\label{tab:para}
\end{table*}

\section{Convergence analysis}
First, we denote $\bar{\mathbf{w}}(t):=\frac{1}{N}\sum_{i=1}^N \mathbf{w}_i(t)$ as the consensus global representation, and for the gradient of the global loss function $L$ with respect to each variable $\mathbf{w}$ and $\textbf{v}$, we have: 
\begin{align}\label{eq:partial}
&\nabla_{\mathbf{w}}L(\bar{\mathbf{w}}(t),\{\textbf{v}_i(t)\}_{i=1}^N):=\frac{1}{N}\sum_{i=1}^N\nabla_{\mathbf{w}}L_i(\bar{\mathbf{w}}(t),\!\textbf{v}_i(t)),\nonumber\allowdisplaybreaks\\
&\nabla_{\textbf{v}}L(\bar{\mathbf{w}}(t),\{\textbf{v}_i(t)\}_{i=1}^N)\!:=\!\frac{1}{N}\!\sum_{i=1}^N \nabla_{\textbf{v}}L_i(\bar{\mathbf{w}}(t),\textbf{v}_i(t)), \nonumber
\end{align}
Then we make some basic assumptions that need to be used to prove convergence.

\begin{assumption}[$L$-Lipschitz Continuous Gradient]\label{assumption-lipschitz}
There exists a constant $L>0$, such that 
$\|\nabla_{\mathbf{w}} L_i(\mathbf{w}, \mathbf{v})-\nabla_{\mathbf{w}} L_i(\mathbf{w}^\prime, \mathbf{v})\| \leq L\|\mathbf{w}-\mathbf{w}^\prime\|$ and
$\|\nabla_{\textbf{v}} L_i(\mathbf{w},\textbf{v})-\nabla_{\textbf{v}} L_i(\mathbf{w},\textbf{v}^\prime)\| \leq L\|\textbf{v}-\textbf{v}^\prime\|, \forall \mathbf{w}, \mathbf{w}^\prime\in{\Phi}, \textbf{v}, \textbf{v}^\prime\in{\Theta}.$
\end{assumption}

\begin{assumption}[Unbiased Local Gradient Estimator]\label{assumption-gradient}
The local gradient estimators are unbiased, i.e., $\forall \mathbf{w}, \mathbf{w}^\prime\in{\Phi}, \textbf{v}, \textbf{v}^\prime\in{\Theta}.$ $\forall i\in[N],$ 
    $\mathbb{E}[g_{\mathbf{w}}(\mathbf{w}_i,\mathbf{v}_i)]=\nabla_{\mathbf{w}} L_i(\mathbf{w}_i,\mathbf{v}_i),~
    \mathbb{E}[g_{\mathbf{v}}(\mathbf{w}_i,\mathbf{v}_i)]=\nabla_{\mathbf{v}} L_i(\mathbf{w}_i,\mathbf{v}_i)$
with the expectation being taken over the local data samples.
\end{assumption}

\begin{assumption}[Bounded Variance]\label{assumption-variance}
{There exists a constant $\sigma>0$ such that the variance of each local gradient estimator is bounded}, i.e.,  $\forall \mathbf{w}, \mathbf{w}^\prime\in{\Phi}, \textbf{v}, \textbf{v}^\prime\in{\Theta}.$ $\forall i\in[N],$ 
    $\mathbb{E}[\|g_{\mathbf{w}}(\mathbf{w}_i,\mathbf{v}_i)-\nabla_{\mathbf{w}} L_i(\mathbf{w}_i,\mathbf{v}_i)\|^2]\leq \sigma_1^2, ~
    \mathbb{E}[\|g_{\mathbf{v}}(\mathbf{w}_i,\mathbf{v}_i)-\nabla_{\mathbf{v}} L_i(\mathbf{w}_i,\mathbf{v}_i)\|^2]\leq \sigma_2^2.$
\end{assumption}

\begin{assumption}[Bounded Global Variability]\label{assumption:global-var}
 There exists a constant  $\varsigma>0$ such that the global variability of the local partial gradients on $\mathbf{w}$ of the loss function $\forall \mathbf{v}_i\in\Theta$ is bounded, i.e., 
   $\frac{1}{N}\sum_{i=1}^N \mathbb{E}[\|\nabla_{\mathbf{w}} L_i(\mathbf{w},\mathbf{v}_i)-\nabla_{\mathbf{w}} L(\mathbf{w},\{\mathbf{v}_i\}_{i=1}^N)\|^2]\leq\varsigma^2.$
\end{assumption}

We will show that the following expression is bounded above by $\epsilon$. This demonstrates that the composite term remains bounded, indicating that our algorithm converges under certain conditions and eventually stabilizes within a constant range:
\begin{equation}\label{eq:lyapunov-function3}
\begin{aligned}
M(t)&:=\underset{\text{consensus error of global representation}~\mathbf{w}}{\underbrace{\frac{1}{N}\sum_{i=1}^N\|\mathbf{w}_i(t)-\bar{\mathbf{w}}(t)\|^2}}\nonumber\\
+&\underset{\text{partial gradient w.r.t.}~\mathbf{w}~\text{of global loss function}}{\underbrace{\left\|\nabla_{\mathbf{w}}L(\bar{\mathbf{w}}(t),\{\mathbf{v}_i(t+1)\}_{i=1}^N)\right\|^2}}\nonumber\\
+&\frac{\alpha\tau}{\beta}\!\!\!\!\!\underset{\text{partial gradient w.r.t.}~\mathbf{v}~\text{of global loss function}}{\underbrace{\left\|\nabla_{\mathbf{v}}L(\bar{\mathbf{w}}(t),\{\mathbf{v}_i(t)\}_{i=1}^N)\right\|^2}}\nonumber
\end{aligned}
\end{equation}

The first term in Eq.~\ref{eq:lyapunov-function3} measures the average error of the learned global knowledge representation from the perspective of the shared adapter of each agent $a_i$. The last two terms represent the performance of our model. We use different learning rates $\eta_w$ and $\eta_v$ for the shared adapter and the personalized adapter. We are inspired by finite-time analysis of two-timescale stochastic approximation \cite{borkar2008stochastic} to consider the global loss of weighted partial gradients \cite{xiong2024deprl}.

Our proof is based on Alg.~\ref{alg:PE-MA1}, where the local update alternates between updating the personalized adapter and the shared adapter.
Meanwhile, Alg.~\ref{alg:PE-MA2}, as a variant of Alg.~\ref{alg:PE-MA1}, simultaneously updates both the personalized adapter and the shared adapter in the same round, which has been mentioned in the Appendix. Based on this, we derive the following theoretical results:

\begin{theorem}\label{thm:loss_convergence}
Under Assumptions \ref{assumption-lipschitz}-\ref{assumption:global-var}, we choose learning rates satisfying $\alpha\leq \frac{1+36\tau^2}{\tau L}$ and $\beta\leq \min\left(1/L,N/2,\frac{1-q}{3\sqrt{2}CLN}\right)$, where $C:=\frac{2(1+p^{-N})}{1-p^{N}}$, $q:=(1-p^{N})^{1/N}$ and $p=\arg\min P_{ij}, \forall i,j, P_{ij}>0.$ 
 Denote the optimal parameters of shared adapter and personalized adapter as $\mathbf{w}^*$ and $\{\mathbf{v}_i^*\}_{i=1}^N$, respectively.  The sequence of parameters  $\{ \{\mathbf{w}_i(k)\}_{i=1}^N,  \{\mathbf{v}_i(k)\}_{i=1}^N,\forall k\}$ generated by \emph{PE-MA} satisfy
\begin{align}\label{eq:thm_loss}
&\frac{1}{K}\!\sum\limits_{k=0}^{K-1}\mathbb{E}[M(t)]\nonumber\displaybreak[1]\\
&\leq\frac{4L(\bar{\mathbf{w}}(0),\!\{\mathbf{v}_i(0)\}_{i=1}^N)\!-\!4 L({\mathbf{w}^*},\!\{\mathbf{v}_i^*\}_{i=1}^N)}{K\eta_\mathbf{w}}\nonumber\displaybreak[1]\\
&+\frac{2\eta_\mathbf{w} L}{N}\sigma_1^2+\frac{12\eta_\mathbf{v}^3L^2\tau}{\eta_\mathbf{w}}(\tau-1)(6\tau-1)\sigma_1^2+\frac{2\eta_\mathbf{v}^2\tau L}{\eta_\mathbf{w}}\sigma_2^2\nonumber\displaybreak[1]\\
&+\frac{2\eta_\mathbf{w}}{3 N}\!\left(\!1\!+\!\frac{1}{L^2}\!\right)\!\sigma_1^2\!+\!\frac{2\eta_\mathbf{w}}{N}\!\left(\!1\!+\!\frac{1}{L^2}\!\right)\!\varsigma^2.
\end{align}
\end{theorem}


\begin{remark}
    The term $\frac{4L(\bar{\mathbf{w}}(0),\!\{\mathbf{v}_i(0)\}_{i=1}^N)\!-\!4 L({\mathbf{w}^*},\!\{\mathbf{v}_i^*\}_{i=1}^N)}{K\eta_w}$ vanishes at a rate proportional to $1/K$, reflecting the initial performance gap between the algorithm and the optimal solution, which gradually diminishes with training. This convergence behavior aligns with that observed in traditional decentralized learning frameworks \cite{assran2019stochastic, lian2017can}. In addition, there exists a persistent error term that does not decrease as $K$ increases, representing the long-term “steady-state error” of the algorithm. These errors mainly arise from two sources: (i) variance due to stochastic gradients and model heterogeneity, represented by $\sigma_1, \sigma_2$ and $\varsigma$; and (ii) the discrepancy between each agent’s shared adapter and the average shared adapter, measured by $\|\mathbf{w}_i(k)-\bar{\mathbf{w}}(k)\|^2$. Notably, the term $\frac{12\eta_v^3L^2\tau}{\eta_w}(\tau-1)(6\tau-1)\sigma_1^2+\frac{2\eta_v^2\tau L}{\eta_w}\sigma_2^2$ originates from the accumulation of local update errors during the first step. This indicates that each round of local updates progressively introduces errors, and that these errors grow at a higher-order rate with respect to the number of local updates $\tau$.
\end{remark}

At the same time, we can also easily get a concise result by taking the values of $\eta_w$ and $\eta_v$:
\begin{corollary}\label{cor:1}
Let $\eta_v=\frac{1}{\tau\sqrt{K}}$ and $\eta_w=\sqrt{{N}/{K}}$. The convergence rate of \emph{PE-MA} is 
$\mathcal{O}\left(\frac{1}{\sqrt{NK}}+\frac{1}{K\sqrt{N}}+\frac{1}{\tau\sqrt{NK}}\right),$ 
when the total number of communication rounds $K$ satisfies 
    $K\geq 
    \max\left(\frac{18C^2L^2N^3}{(1-q)^2},\frac{(2L^2+2)^2}{NL^4},NL^2\right).$
\end{corollary}

\section{Experiment}
\subsection{Dataset}
We systematically evaluate the proposed multi-agent collaborative framework on three widely used cross-domain image classification benchmark datasets, namely Office-Home\cite{venkateswara2017deep}, Office-Caltech\cite{gong2012geodesic}, DomainNet\cite{peng2019moment}.
\begin{itemize}
    \item \textbf{Office-Home} is a medium-sized cross-domain dataset covering four subdomains with different styles: Art, Clipart, Product, and Real-World, with more than $15,500$ images and $65$ categories.
    \item  \textbf{Office-Caltech} is a classic small-scale cross-domain dataset consisting of four subdomains: Amazon, Webcam, DSLR, and Caltech-256, whose categories are the overlapped parts of Office31 and Caltech-256 (a total of $10$ categories).
    \item \textbf{DomainNet} is one of the largest cross-domain image classification datasets, including six subdomains (Clipart, Infograph, Painting, Quickdraw, Real, and Sketch) and $345$ categories, totaling about $600,000$ images.
\end{itemize}
 
\subsection{Experimental Setup}

\textbf{Evaluation Metrics.}
We adopt \emph{classification accuracy} as the primary evaluation metric. To further assess the impact of our method on individual agent performance, we also report the classification accuracy of each agent on its local domain. This per-agent evaluation is particularly important in personalized multi-agent learning settings, as it reflects the model’s adaptability and fairness across heterogeneous agents.

In addition, we compare the number of trainable parameters and the amount of communicated parameters under different model architectures and update strategies. These metrics act as proxies for assessing communication bandwidth demands and local memory usage, both of which are critical considerations in real-world multi-agent distributed deployment scenarios.

\textbf{Settings and Implementation Details.} Our method is general and can be used in different collaborative learning frameworks. We adopt ResNet-18 as the default backbone network and embed our designed dual-adapter structure.
One adapter is used for local personalized modeling. The other is used for global knowledge sharing. This design improves both personalization and collaboration across agents. We also test the method under different network topologies to show its flexibility and robustness. To handle sudden agent dropouts, we design an aggregation method with dual randomness. This helps the method adapt better in low-resource and unstable environments.

To verify the effectiveness and communication efficiency of the proposed method in multi-agent collaborative image classification tasks, we conduct experiments under a setting that closely reflects real-world scenarios. By default, each agent only has access to a small portion of labeled data within its own domain (with a default ratio of $80\%$), and the number of samples for the same category varies significantly across different domains. 

For experimental comparison, we consider several representative methods, including standard full-model training and personalized collaborative learning strategies based on FedSim \cite{Liang2020ThinkLA} and FedAlt \cite{collins2021exploiting}, applied to three different levels: input layer, outer layer, and adapter layer. During training, the backbone network parameters remain frozen, and only the adapter parameters, encoding head, or fully connected layer parameters are updated depending on the specific method used. The initial learning rate is set to $10^{-2}$, and it becomes $10^{-1}$ of the previous one every three epochs, and the batch size is 64. We implement all algorithms in PyTorch \cite{paszke2017automatic} on Python 3 with ten NVIDIA RTX 3090 GPUs.

\subsection{Result}
In Tab.~\ref{tab:main-results} and~\ref{tab:para}, we compare the different effects of different methods on image classification tasks. Based on these results, we conducted an analysis.

\textbf{Accuracy} In Tab.~\ref{tab:main-results}, We systematically evaluate the performance of the proposed \emph{PE-MA} method on three cross-domain image classification benchmark datasets: Office-Home, Office-Caltech10, and DomainNet, under three representative communication topologies: FC, ER, and Ring. 

From the dataset perspective: 

\begin{itemize}
    \item On the Office-Home dataset, \emph{PE-MA} achieves the best performance under all three communication topologies, reaching the highest accuracy of {82.32\%} under the ER topology. This significantly outperforms the standalone full model (80.53\%) and all DSGD-SIM/ALT variants, verifying the robustness and collaborative efficiency of the dual-adapter structure on medium-scale heterogeneous tasks.
    \item On the Office-Caltech10 dataset, \emph{PE-MA} also maintains superior performance under all communication topologies, with the highest accuracy reaching {97.79\%}. This demonstrates that in small-scale and less heterogeneous scenarios, the personalized adapter in \emph{PE-MA} can effectively leverage local data advantages while preserving strong collaborative capabilities.
    \item On the large-scale and highly heterogeneous DomainNet dataset, \emph{PE-MA} shows outstanding cross-domain generalization ability, outperforming existing methods under all topologies, with the highest accuracy of {64.47\%} (Ring topology). Compared with full model training (60.13\%) and DSGD-SIM Adapter (63.57\%), \emph{PE-MA} shows clear improvements, demonstrating its capability in effective knowledge transfer across agents.
    \item Among the three benchmark datasets, \emph{PE-MA} shows the most prominent improvement on the DomainNet dataset, with an accuracy increase of approximately $4\%$, from $60.13\%$ to $64.08\%$. Given DomainNet’s large scale and high domain heterogeneity, it presents more significant challenges for cross-domain generalization. The dual-adapter design of \emph{PE-MA} effectively distinguishes between local personalization and global knowledge sharing, thereby improving cross-domain transferability.
\end{itemize}

From the topology perspective:

\begin{itemize}
    \item In the FC topology, where communication is maximized, \emph{PE-MA} achieves optimal or near-optimal performance across all datasets. This indicates that with full communication, \emph{PE-MA} effectively integrates local and global knowledge.
    \item In the ER topology, although sparser, the graph remains connected. \emph{PE-MA} often performs best under this topology, reaching the highest global accuracy of {82.32\%} on Office-Home. This suggests that even under constrained communication, \emph{PE-MA} maintains strong collaborative learning capabilities.
    \item In the Ring topology, which offers the least communication and tests the method’s resilience under resource-constrained conditions, \emph{PE-MA} still delivers remarkable performance. Notably, it achieves {64.47\%} accuracy on the DomainNet dataset, outperforming all baselines. This highlights \emph{PE-MA}'s robustness and adaptability in low-communication environments.
    \item Among the three communication topologies, \emph{PE-MA} shows particularly strong robustness and performance under the Ring topology, achieving an accuracy of $64.08\%$, which is approximately $2\%$ to $4\%$ higher than the other methods (such as DSGD-SIM and DSGD-ALT). Due to the longer communication paths and lower redundancy in the Ring structure, it requires higher fault tolerance. \emph{PE-MA}’s dual-randomness aggregation mechanism demonstrates strong adaptability in such constrained environments, effectively mitigating performance degradation caused by limited communication.
\end{itemize}

\textbf{Parameter quantity} Tab.~\ref{tab:para} summarizes the communication and training parameter sizes of different personalization strategies across three datasets. 

\begin{itemize}
    \item In terms of communication overhead, \emph{PE-MA} on all datasets is significantly lower than that of the DSGD series of methods. Taking the Office-Home dataset as an example, the communication parameter volume of \emph{PE-MA} is 1.44M, which is about $77.8\%$ less than DSGD-ALT (Input Layer update, 6.48M) and about $87.2\%$ less than DSGD-SIM (Adapter update, 11.21M). This significant reduction in communication volume is attributed to the lightweight shared adapter synchronization mechanism adopted by \emph{PE-MA}, which only needs to synchronize a small number of shared parameters, greatly alleviating the communication bottleneck in the multi-agent collaborative environment.
    \item In terms of local training load, \emph{PE-MA} also shows great advantages. Taking DomainNet as an example, the training parameter volume of \emph{PE-MA} is 3.18M, which is about $72.0\%$ less than DSGD-ALT (Input)'s 11.35M; it is about $75.1\%$ less than DSGD-SIM (Adapter)'s 12.75M. On Office-Home and Office-Caltech10, it can also be observed that the training parameter volume generally decreases by more than $70\%$. This shows that \emph{PE-MA} significantly reduces the local computing and storage burden by freezing the backbone network and training only for lightweight adapters.
    \item Specifically, the {Indep} method does not require communication; however, it cannot leverage cross-agent information sharing, resulting in an overall accuracy significantly lower than that of \emph{PE-MA}. In contrast, \emph{PE-MA} achieves substantial performance improvements with only a minimal increase in communication and training overhead---less than $25\%$ of that required by traditional methods---demonstrating an excellent balance between efficiency and effectiveness. 
    \item Overall, the dual-adapter architecture of \emph{PE-MA} successfully separates global knowledge sharing from local personalized learning, enabling a well-balanced trade-off among model performance, communication efficiency, and computational cost.
\end{itemize}

\subsection{Ablation Study}
\begin{figure}[h]
	\vskip -0.1in
	\begin{center}
		\includegraphics[width=1.0\columnwidth]{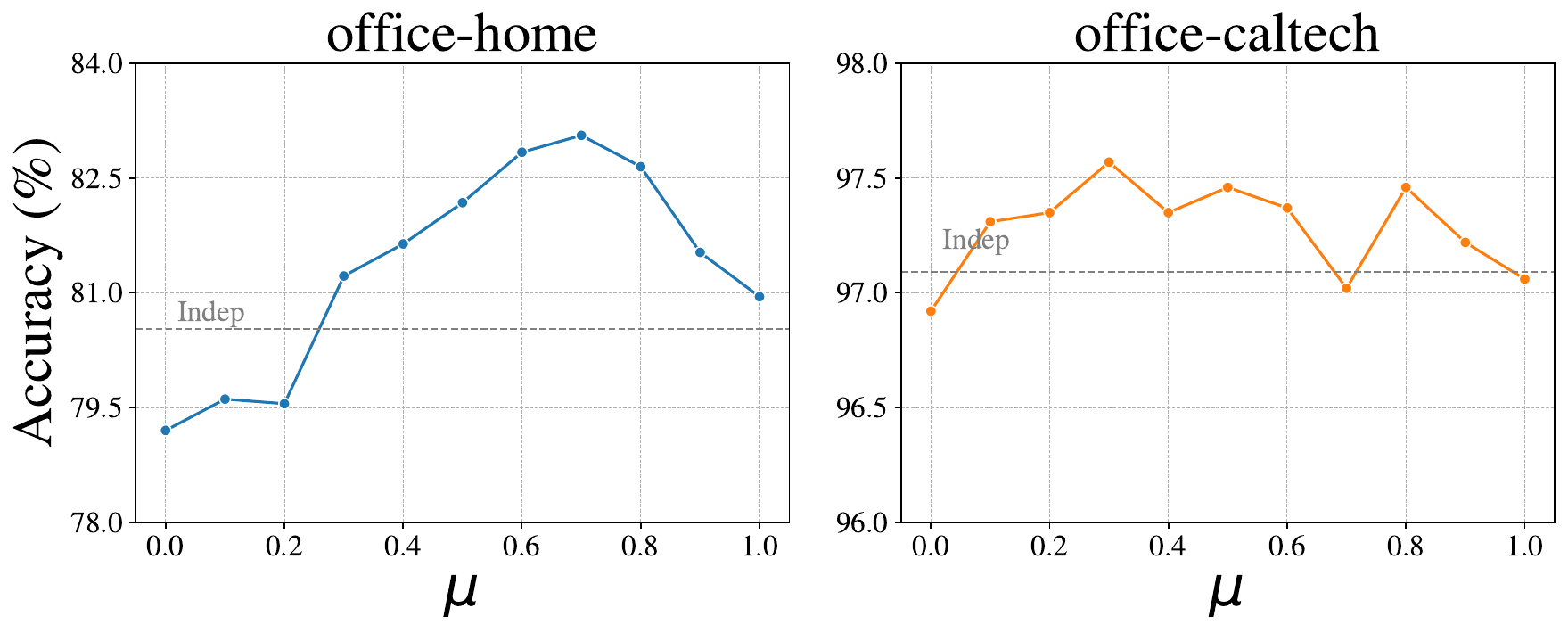}
		\caption{Peak accuracy of Office-Home10 and Office-Caltech10 at different personalization levels($\mu$)}
		\label{fig:oh10&oc10-mu}
	\end{center}
	\vskip -0.2in
\end{figure}
\textbf{Tradeoff between global and personalized adapter} 
As shown in Fig.~\ref{fig:oh10&oc10-mu}, the weight of the personalized adapter has a significant impact on the best accuracy the model can achieve on both datasets. On the Office-Home10 dataset, as the personalization weight increases, the model's accuracy first improves and then declines, indicating that moderately enhancing personalization helps extract features tailored to specific tasks. However, an excessively high weight may hinder the learning of knowledge from other agents, potentially leading to overfitting. In contrast, the accuracy on the Office-Caltech10 dataset remains relatively stable, suggesting either a lower dependence on personalization or that the model's generalization capability is already sufficient. Therefore, the trade-off between personalization and shared knowledge plays a crucial role in determining final model performance across different tasks.

\begin{figure}[h]
	\begin{center}
		\includegraphics[width=0.8\columnwidth]{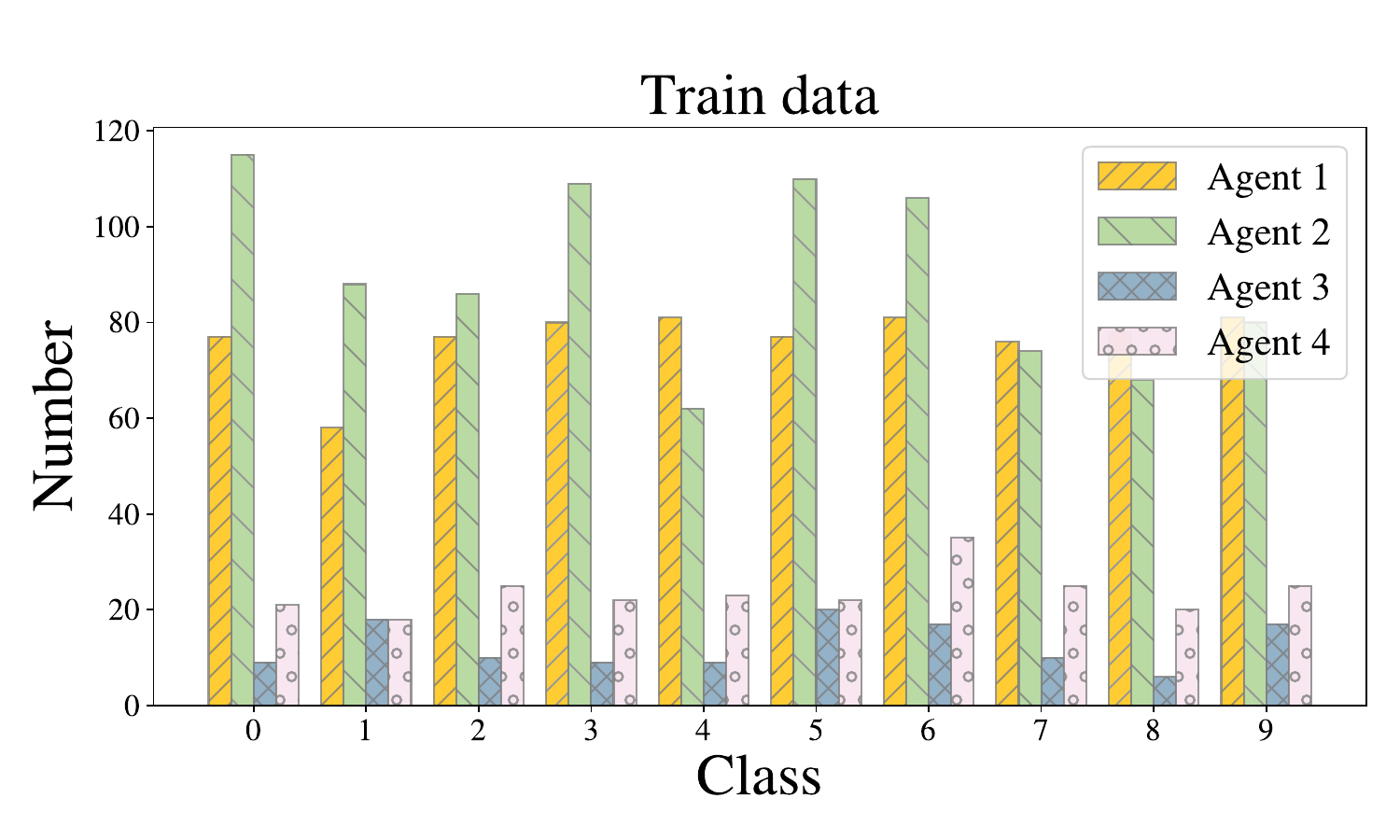}
		\caption{Different data distributions formed by dividing the Office-Caltech10 dataset into four agents}
		\label{fig:dataset-oc10-distribution}
	\end{center}
	\vskip -0.1in
\end{figure}

\begin{figure}[b]
	\vskip -0.2in
	\begin{center}
		\includegraphics[width=1.0\columnwidth]{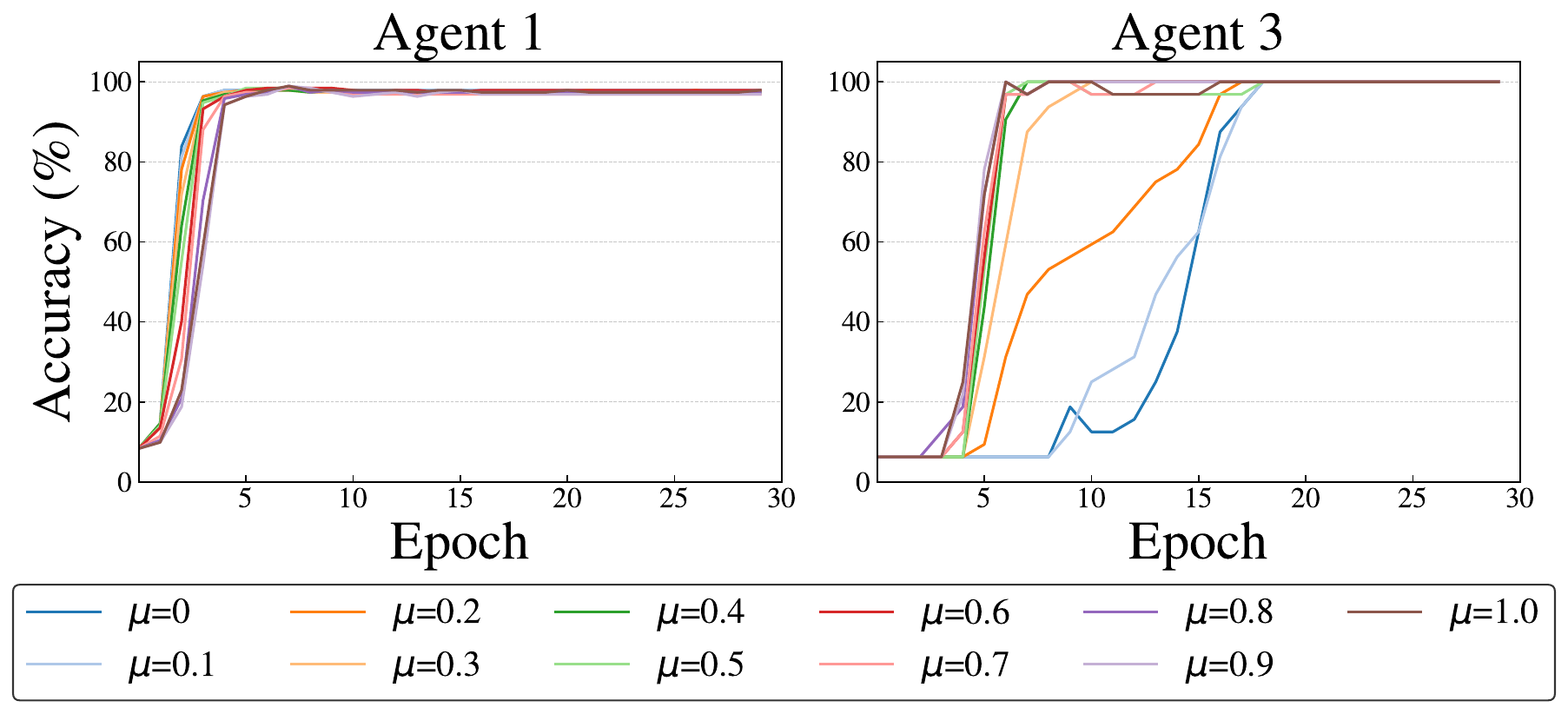}
		\caption{The accuracy of two agents with different data distributions performing tasks at different personalization levels($\mu$)}
		\label{fig:oh10-clients-mu}
	\end{center}
	\vskip -0.2in
\end{figure}

{\textbf{Data-Responsive Personalization Design} }
As illustrated in Fig.~\ref{fig:oh10-clients-mu}, under the data distribution shown in Fig.~\ref{fig:dataset-oc10-distribution}, the personalization weight $\mu$ has different impacts across clients. For Client 1, which has sufficient data, changes in $\mu$ have minimal effect on convergence or accuracy. However, for Client 3, which has limited data, reducing $\mu$—i.e., increasing reliance on the personalized adapter—leads to faster convergence and better final performance. These findings suggest that clients with ample data are less sensitive to personalization weights, while clients with sparse or heterogeneous data benefit significantly from adjusting the balance between global and personalized adapters. Proper tuning of $\mu$ helps the model better fit local data and improves overall efficiency. 

Next, we will further explore the personalized strategy for specific agents. Fig.~\ref{fig:oh10-adaptivemu} shows the change in accuracy of four agents (Office-Home is divided into four different domains) during training in a decentralized multi-agent learning scenario. The three curves represent the three strategies: training each agent separately, fixed personalization, and adaptive personalization. In general, adaptive personalization can achieve the most effective learning. This highlights the importance of dynamically adjusting the personalization based on data characteristics. Agent 1 and agent 2 have similar changes in accuracy. They have sufficient local data, so the accuracy increases and stabilizes at a similar rate under the three strategies. Agent 3 and agent 4 have less local data and uneven distribution, so the gap between the three strategies can be clearly seen. Using adaptive personalization will obviously reach a stable accuracy rate $30\%-50\%$ faster than using fixed personalization. The gap between the three strategies further emphasizes the advantage of using adaptive personalization strategies in situations where the amount of data or data distribution is uneven.
\begin{figure}[h]
	\vskip -0.1in
	\begin{center}
		\includegraphics[width=1.0\columnwidth]{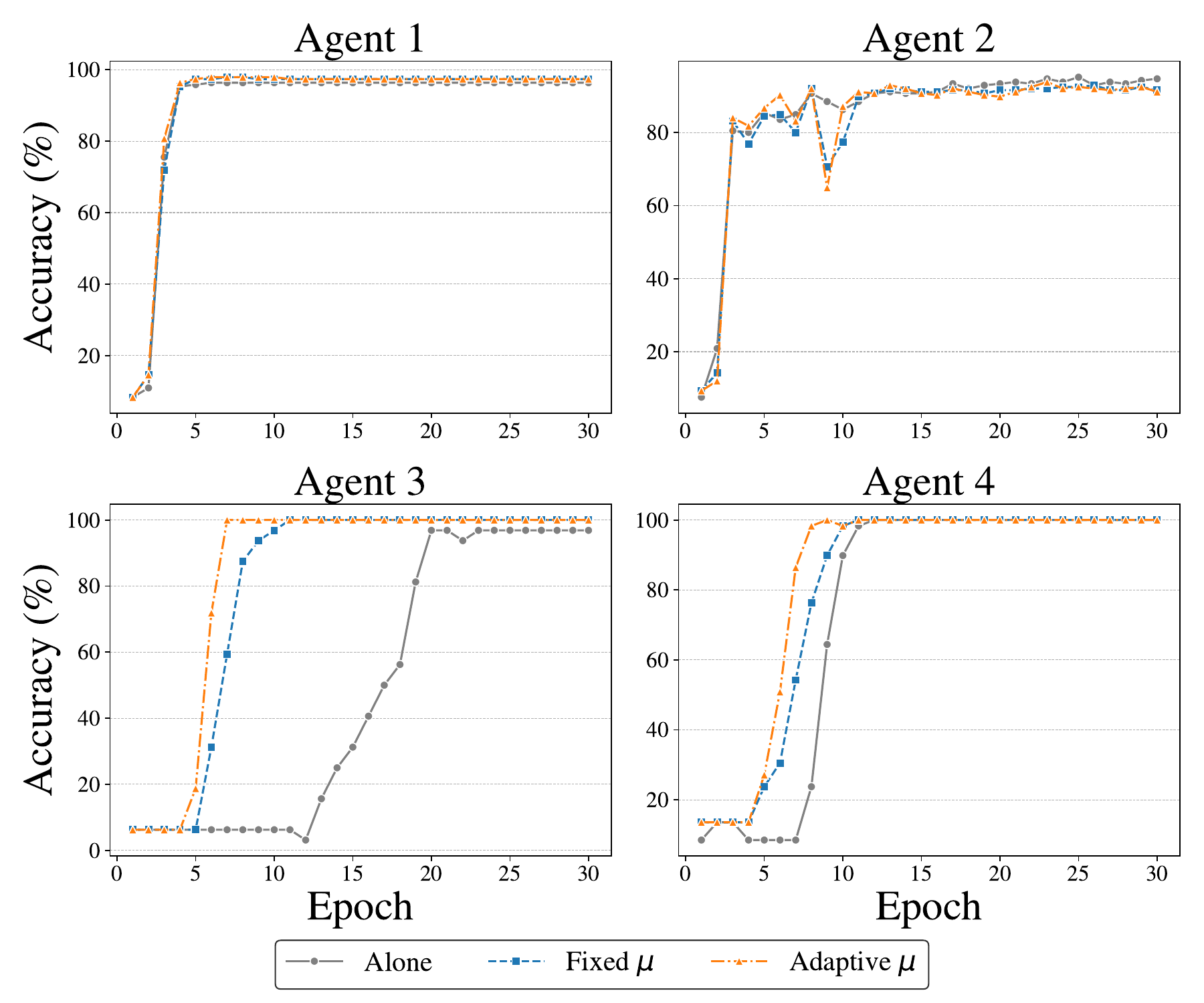}
		\caption{The accuracy of each agent using the same personalization level and different personalization levels according to the data volume and data distribution}
		\label{fig:oh10-adaptivemu}
	\end{center}
	\vskip -0.2in
\end{figure}

\textbf{Impact of Communication Stability} 
Fig.~\ref{fig:oh10-offline} illustrates the impact of communication stability on model performance under a decentralized multi-agent framework. On the Office-Home dataset (left), the model converges stably regardless of whether all nodes maintain stable communication (solid line) or experience a $10\%$ probability of disconnection (dashed line). However, the accuracy under stable communication is slightly higher, indicating that communication integrity plays a positive role in enhancing final model performance. A similar trend is observed on the Office-Caltech dataset (right), where the model also performs better under stable communication conditions. 
Compared to the baseline without any communication, both communication settings significantly improve the model’s accuracy, demonstrating that information exchange among agents is crucial for performance enhancement. Overall, although unstable communication may introduce some performance fluctuations, a well-designed communication mechanism can still ensure robust training and high accuracy in heterogeneous environments—especially in the later stages of training.

\begin{figure}[h]
	\vskip -0.1in
	\begin{center}
		\includegraphics[width=1.0\columnwidth]{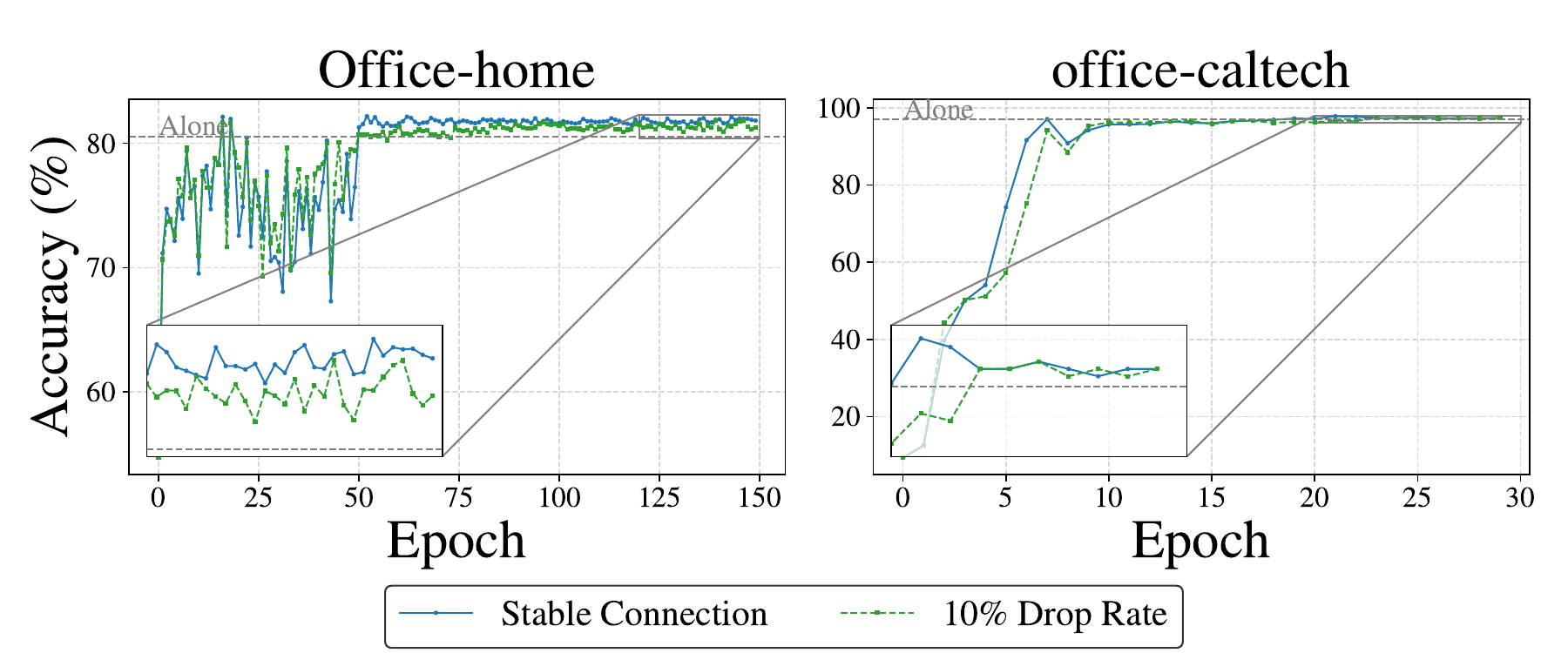}
		\caption{Accuracy comparison under agent dropout ($10\%$) vs no dropout ($0\%$).}
		\label{fig:oh10-offline}
	\end{center}
	\vskip -0.2in
\end{figure}

\section{Conclusion}
This paper proposes the \emph{PE-MA} method, which combines shared adapters with personalized adapters to achieve efficient collaboration and local personality optimization in a multi-agent environment. {We prove that \emph{PE-MA} can achieve convergence results in a decentralized framework by considering both the error of global knowledge representation and model performance, and that the convergence rate can reach $\mathcal{O}\left(\frac{1}{\sqrt{NK}}+\frac{1}{K\sqrt{N}}+\frac{1}{\tau\sqrt{NK}}\right),$ when the learning rate $\eta_w$ and $\eta_v$ reaches a specific value.} A large number of experiments show that \emph{PE-MA} can maintain excellent classification performance on a variety of cross-domain datasets while significantly reducing communication overhead and training burden. Further ablation analysis verifies the impact of adapter weight settings on convergence rate and performance, especially in scenarios with less data. Overall, \emph{PE-MA} provides an efficient, stable and scalable solution for multi-agent collaborative learning tasks. In future work, we plan to further enhance the adaptability of \emph{PE-MA} by introducing dynamic adapter selection strategies and exploring its integration with large-scale foundation models in real-world edge applications.

\bibliography{IEEEabrv, example_paper}
\bibliographystyle{IEEEtran}

\clearpage
\setcounter{page}{1}
\appendices
\section{Convergence}

We have parameters of shared adapter and personalized adapter. According to the conclusion of the Lipschitz gradient continuous function in Assumption \ref{assumption-lipschitz}, it is not difficult to deduce:
\begin{lemma}\label{lemma:2D-Lipschitz}
Under Assumption \ref{assumption-lipschitz}, we have 
\begin{align*}
    &L_i(\mathbf{w}, \mathbf{v})-L_i(\mathbf{w}^\prime, \mathbf{v}^\prime)\\
   \leq \quad & \nabla_{\mathbf{w}}L_i(\mathbf{w}^\prime,\mathbf{v})^\intercal (\mathbf{w}-\mathbf{w}^\prime)\\
   +\quad &\frac{L}{2}\|\mathbf{w}-\mathbf{w}^\prime\|^2+\nabla_{\mathbf{v}}L_i(\mathbf{w}^\prime,\mathbf{v}^\prime)^\intercal (\mathbf{v}-\mathbf{v}^\prime)+\frac{L}{2}\|\mathbf{v}-\mathbf{v}^\prime\|^2. 
\end{align*}
\end{lemma}

\begin{proof}
We rewrite $L_i(\mathbf{w}, \mathbf{v})-L_i(\mathbf{w}^\prime, \mathbf{v}^\prime)$ as follows:
\begin{align*}
   &L_i(\mathbf{w}, \mathbf{v})-L_i(\mathbf{w}^\prime, \mathbf{v}^\prime)\\
    =\quad & L_i(\mathbf{w}, \mathbf{v})- L_i(\mathbf{w}^\prime, \mathbf{v})+L_i(\mathbf{w}^\prime, \mathbf{v})-L_i(\mathbf{w}^\prime,\mathbf{v}^\prime)\\
    \leq\quad& \nabla_{\mathbf{w}}L_i(\mathbf{w}^\prime,\mathbf{v})^\intercal (\mathbf{w}-\mathbf{w}^\prime)+\frac{L}{2}\|\mathbf{w}-\mathbf{w}^\prime\|^2 \\
    &+\nabla_{\mathbf{v}}L_i(\mathbf{w}^\prime,\mathbf{v}^\prime)^\intercal (\mathbf{v}-\mathbf{v}^\prime)+\frac{L}{2}\|\mathbf{v}-\mathbf{v}^\prime\|^2,
\end{align*}
where the inequality is due to the fact \cite{bottou2018optimization} that  $\|\nabla_{\mathbf{w}} F_j(\mathbf{w},\mathbf{v})-\nabla_{\mathbf{w}} F_j(\mathbf{w}^\prime,\mathbf{v})\|\leq L\|\mathbf{w}-\mathbf{w}^\prime\|$ and
$\|\nabla_{\mathbf{v}} F_j(\mathbf{w},\mathbf{v})-\nabla_{\mathbf{v}} F_j(\mathbf{w},\mathbf{v}^\prime)\|\leq L\|\mathbf{v}-\mathbf{v}^\prime\|$.
\end{proof}

To prove Theorem \ref{thm:loss_convergence}, we start with evaluating the drift of the global loss function, i.e., $L(\bar{\mathbf{w}}(t+1),\{\mathbf{v}_i(t+1)\}_{i=1}^N)-L(\bar{\mathbf{w}}(t),\{\mathbf{v}_i(t)\}_{i=1}^N)$. 
\begin{lemma}
The drift of the global loss function satisfies
\begin{align}\label{eq:global_iter2}
  & \mathbb{E}[L(\mathbf{w}(t+1),\{\mathbf{v}_i(t+1)\}_{i=1}^N)]- \mathbb{E}[L(\mathbf{w}(t),\{\mathbf{v}_i(t)\}_{i=1}^N)]\nonumber\\
   \leq \quad&\underset{C_1}{\underbrace{\frac{1}{N}\sum_{i=1}^N\mathbb{E}\Big\langle \nabla_{w}L_i(\mathbf{w}(t),\mathbf{v}_i(t+1)), \mathbf{w}(t+1)-\mathbf{w}(t)\Big\rangle}}\nonumber\\
   &+\underset{C_2}{\underbrace{\frac{1}{N}\sum_{i=1}^N\frac{L}{2}\mathbb{E}[\|\mathbf{w}(t+1)-\mathbf{w}(t)\|^2]}}\nonumber\\
   &+\underset{C_3}{\underbrace{\frac{1}{N}\sum_{i=1}^N\mathbb{E}\Big\langle \nabla_{\mathbf{v}}L_i(\mathbf{w}(t),\mathbf{v}_i(t)), \ \mathbf{v}_i(t+1) -\ \mathbf{v}_i(t) \Big\rangle}}\nonumber\\
   &+\underset{C_4}{\underbrace{\frac{1}{N}\sum_{i=1}^N\frac{L}{2}\mathbb{E}[\|{\mathbf{v}_i}(t+1)-{\mathbf{v}_i}(t)\|^2]}}.  
\end{align}
\end{lemma}
\begin{proof}
According to the definition of the global loss function $f$, we have 
\begin{align}\label{eq:global_iteration}\nonumber
    &\mathbb{E}[L(\mathbf{w}(t+1),\{\mathbf{v}_i(t+1)\}_{i=1}^N)]\\
    =\quad&\frac{1}{N}\sum_{i=1}^N \mathbb{E}[L_i(\mathbf{w}(t+1),\mathbf{v}_i(t+1))]\nonumber \allowdisplaybreaks\\
    \leq \quad& \frac{1}{N}\sum_{i=1}^N \mathbb{E}[\nabla_{w}L_i(\mathbf{w}(t),\mathbf{v}_i(t+1))^\intercal (\mathbf{w}(t+1)-\mathbf{w}(t))]\nonumber\\
    &+\frac{1}{N}\sum_{i=1}^N\frac{L}{2}\mathbb{E}[\|\mathbf{w}(t+1)-\mathbf{w}(t)\|^2]\nonumber\allowdisplaybreaks\nonumber\\
    &+\frac{1}{N}\sum_{i=1}^N \mathbb{E}[\nabla_{\mathbf{v}}L_i(\mathbf{w}(t),\mathbf{v}_i(t))^\intercal ({\mathbf{v}}_i(t+1)-{\mathbf{v}_i}(t))]\nonumber\\
    &+\frac{1}{N}\sum_{i=1}^N\frac{L}{2}\mathbb{E}[\|{\mathbf{v}}_i(t+1)-{\mathbf{v}}_i(t)\|^2]\nonumber\allowdisplaybreaks\\
    &+\frac{1}{N}\sum_{i=1}^N \mathbb{E}[L_i(\mathbf{w}(t),\mathbf{v}_i(t))],
\end{align}
where the inequality follows from Lemma \ref{lemma:2D-Lipschitz}. The desired result holds by moving $\frac{1}{N}\sum_{i=1}^N \mathbb{E}[L_i(\mathbf{w}(t),\mathbf{v}_i(t))]$ into the left-hand side of the inequality and replacing it by $\mathbb{E}[l(\mathbf{w}(t),\{\mathbf{v}_i(t)\}_{i=1}^N)]$. 
\end{proof}
In the following, we bound  $C_1, C_2, C_3,$ and $C_4$, respectively. 

\begin{lemma}\label{lem:C1}
$C_1$ can be bounded as follows:
\begin{align}
    C_1\leq& {\frac{\eta_w L^2}{2N}\sum_{i=1}^N\mathbb{E}\left[\left\|\bar{\mathbf{w}}(t)-\mathbf{w}_i(t)\right\|^2\right]}\nonumber\\
    &-\frac{\eta_w}{2}\left\|\nabla_{\mathbf{w}}L(\bar{\mathbf{w}}(t),\{\mathbf{v}_i(t+1)\}_{i=1}^N)\right\|^2\nonumber\allowdisplaybreaks\\
    &-\frac{\eta_w}{2N^2}\mathbb{E}\left[\left\|\sum_{i=1}^N g_{\mathbf{w}}(\mathbf{w}_i(t),\mathbf{v}_i(t+1))\right\|^2\right].
\end{align}
\end{lemma}
\begin{proof}
$C_1$ can be bounded as follows
\begin{align}\nonumber
C_1=\quad&\mathbb{E}\Bigg\langle \frac{1}{N}\sum_{i=1}^N\nabla_{\mathbf{w}}L_i(\bar{\mathbf{w}}(t),\mathbf{v}_i(t+1)), \bar{\mathbf{w}}(t+1)-\bar{\mathbf{w}}(t)\Bigg\rangle\nonumber\allowdisplaybreaks\\
    \overset{(a_1)}=\quad&-\eta_w \mathbb{E}\Big\langle \nabla_{\mathbf{w}}L(\bar{\mathbf{w}}(t),\{\mathbf{v}_i(t+1)\}_{i=1}^N),\nonumber\\
    &\qquad\quad\quad\frac{1}{N}\sum_{i=1}^N g_{\mathbf{w}}(\mathbf{w}_i(t),\mathbf{v}_i(t+1))\Big\rangle\nonumber\allowdisplaybreaks\\
    \overset{(a_2)}=\quad&\frac{-\eta_w}{2}\Bigg(\left\|\nabla_{\mathbf{w}}L(\bar{\mathbf{w}}(t),\{\mathbf{v}_i(t+1)\}_{i=1}^N)\right\|^2\\
    &\quad+\frac{1}{N^2}\mathbb{E}\left[\left\|\sum_{i=1}^N g_{\mathbf{w}}(\mathbf{w}_i(t),\mathbf{v}_i(t+1))\right\|^2\right]\nonumber\allowdisplaybreaks\\
    &\quad-\frac{1}{N^2}\mathbb{E}\Bigg[\Bigg\|\sum_{i=1}^N(\nabla_{\mathbf{w}}L_i(\bar{\mathbf{w}}(t),\mathbf{v}_i(t+1))\\
    &\quad-\nabla_{\mathbf{w}}L_i({\mathbf{w}_i}(t),\mathbf{v}_i(t+1)))\Bigg\|^2\Bigg]\Bigg)\nonumber\allowdisplaybreaks \\
    =\quad&\frac{\eta_w}{2N^2}\mathbb{E}\Bigg[\Bigg\|\sum_{i=1}^N(\nabla_{\mathbf{w}}L_i(\bar{\mathbf{w}}(t),\mathbf{v}_i(t+1))\nonumber\\
    &-\nabla_{\mathbf{w}}L_i({\mathbf{w}_i}(t),\mathbf{v}_i(t+1)))\Bigg\|^2\Bigg]\nonumber\allowdisplaybreaks\\
    &-\frac{\eta_w}{2}\left\|\nabla_{\mathbf{w}}L(\bar{\mathbf{w}}(t),\{\mathbf{v}_i(t+1)\}_{i=1}^N)\right\|^2\nonumber\\
    &-\frac{\eta_w}{2N^2}\mathbb{E}\left[\left\|\sum_{i=1}^N g_{\mathbf{w}}(\mathbf{w}_i(t),\mathbf{v}_i(t+1))\right\|^2\right]\nonumber\allowdisplaybreaks\\
    \overset{(a_3)}{\leq}\quad&\underset{D_1}{\underbrace{\frac{\eta_w}{2N}\sum_{i=1}^N\mathbb{E}\Bigg[\Bigg\|(\nabla_{\mathbf{w}}L_i(\bar{\mathbf{w}}(t),\mathbf{v}_i(t+1))}}\nonumber\\
    &\underset{D_1}{\underbrace{-\nabla_{\mathbf{w}}L_i({\mathbf{w}_i}(t),\mathbf{v}_i(t+1)))\Bigg\|^2\Bigg]}}\nonumber\allowdisplaybreaks\\
    &-\frac{\eta_w}{2}\left\|\nabla_{\mathbf{w}}L(\bar{\mathbf{w}}(t),\{\mathbf{v}_i(t+1)\}_{i=1}^N)\right\|^2\nonumber\\
    &-\frac{\eta_w}{2N^2}\mathbb{E}\left[\left\|\sum_{i=1}^N g_{\mathbf{w}}(\mathbf{w}_i(t),\mathbf{v}_i(t+1))\right\|^2\right]\nonumber\allowdisplaybreaks\\
    \overset{(a_4)}{\leq}\quad&{\frac{\eta_w L^2}{2N}\sum_{i=1}^N\mathbb{E}\left[\left\|\bar{\mathbf{w}}(t)-\mathbf{w}_i(t)\right\|^2\right]}\nonumber\\&-\frac{\eta_w}{2}\left\|\nabla_{\mathbf{w}}L(\bar{\mathbf{w}}(t),\{\mathbf{v}_i(t+1)\}_{i=1}^N)\right\|^2\nonumber\allowdisplaybreaks\\
    &-\frac{\eta_w}{2N^2}\mathbb{E}\left[\left\|\sum_{i=1}^N g_{\mathbf{w}}(\mathbf{w}_i(t),\mathbf{v}_i(t+1))\right\|^2\right],
\end{align}
where $(a_1)$ holds due to $\bar{\mathbf{w}}(t+1)-\bar{\mathbf{w}}(t)= \frac{-\eta_w}{N}\sum_{i=1}^N g_{\mathbf{w}}(\mathbf{w}_i(t),\mathbf{v}_i(t+1))$; $(a_2)$ holds according to the equation $\langle a,b\rangle = \frac{1}{2}[\|a\|^2+\|b\|^2-\|a-b\|^2]$ and the fact that $\mathbb{E}\left[g_{\mathbf{w}}(\mathbf{w}_i(t),\mathbf{v}_i(t+1))\right]=\nabla_{\mathbf{w}}L_i(\mathbf{w}_i(t),\mathbf{v}_i(t+1)), \forall i$ from Assumption \ref{assumption-gradient};  $(a_3)$ follows the Cauchy-Schwartz inequality $\|\sum_{i=1}^N a_i\|^2\leq N\sum_{i=1}^N\|a_i\|^2$; and $(a_4)$ directly comes from Assumption \ref{assumption-lipschitz} i.e., $D_1\leq{\frac{\eta_w L^2}{2N}\sum_{i=1}^N\mathbb{E}\left\|\bar{\mathbf{w}}(t)-\mathbf{w}_i(t)\right\|^2}$.
\end{proof}

\begin{lemma}\label{lem:C2}
 $C_2$ is bounded as follows
 \begin{align}
     C_2\leq \frac{\eta_w^2L}{2N}\sigma_1^2+\frac{\eta_w^2L}{2N^2} \left\|\mathbb{E}\sum_{i=1}^N g_{\mathbf{w}}(\mathbf{w}_i(t), \mathbf{v}_i(t+1))\right\|^2.
 \end{align}
\end{lemma}

\begin{proof}
The proof follows from the fact that $\bar{\mathbf{w}}(t+1)-\bar{\mathbf{w}}(t)= \frac{-\eta_w}{N}\sum_{i=1}^N g_{\mathbf{w}}(\mathbf{w}_i(t),\mathbf{v}_i(t+1))$, i.e.,
\begin{align}
   C_2&=\frac{L}{2}\mathbb{E}[\|\bar{\mathbf{w}}(t+1)-\bar{\mathbf{w}}(t)\|^2]\nonumber\\
   &=\frac{\eta_w^2L}{2N^2}\mathbb{E}\left[\left\|\sum_{i=1}^N g_{\mathbf{w}}(\mathbf{w}_i(t), \mathbf{v}_i(t+1))\right\|^2\right]\nonumber\\
   \overset{(b_1)}{=}\quad&\frac{\eta_w^2L}{2N^2}\mathbb{E}\Bigg[\Bigg\|\sum_{i=1}^N g_{\mathbf{w}}(\mathbf{w}_i(t), \mathbf{v}_i(t+1))\nonumber\\
   &-\nabla_{\mathbf{w}}L_i(\mathbf{w}_i(t),\mathbf{v}_i(t+1))\Bigg\|^2\Bigg]\nonumber\\
   &+\frac{\eta_w^2L}{2N^2} \left\|\mathbb{E}\sum_{i=1}^N g_{\mathbf{w}}(\mathbf{w}_i(t), \mathbf{v}_i(t+1))\right\|^2 \nonumber\\
   \overset{(b_2)}{\leq}\quad&\frac{\eta_w^2L}{2N}\sigma_1^2+\frac{\eta_w^2L}{2N^2} \left\|\mathbb{E}\sum_{i=1}^N g_{\mathbf{w}}(\mathbf{w}_i(t), \mathbf{v}_i(t+1))\right\|^2, 
\end{align}
where $(b_1)$ is due to $\mathbb{E}[\|X\|^2]=\mathrm{Var}[X]+\|\mathbb{E}[X]\|^2$; and $(b_2)$ follows from the bounded variance in Assumption \ref{assumption-variance}.
\end{proof}

\begin{lemma}\label{lem:C3}
 $C_3$ can be bounded as
\begin{align}
     C_3{\leq} &\frac{-\eta_v\tau}{2} \mathbb{E}\left[\left\|\nabla_{\mathbf{v}}L(\bar{\mathbf{w}}(t),\{\mathbf{v}_i(t))\}_{i=1}^N\right\|^2\right]\nonumber\\
     &-\frac{\eta_v}{2N\tau}\sum_{i=1}^N\mathbb{E}\left[\left\| \sum_{s=0}^{\tau-1} g_{\mathbf{v}}(\mathbf{w}_i(t),\mathbf{v}_i(t,s))\right\|^2\right]\nonumber\allowdisplaybreaks\\
&+\frac{\eta_v\tau L^2}{N}\sum_{i=1}^N\mathbb{E}\left[\left\|\mathbf{w}_i(t)- \bar{\mathbf{w}}(t)\right\|^2\right]\nonumber\\
&+\eta_v^3L^2(18\tau^3-15\tau^2-3\tau)\sigma^2\nonumber\\
&+\frac{\eta_v^3L^2(18\tau^3-18\tau^2)}{N}\sum_{i=1}^N\mathbb{E}\left[\left\|g_{\mathbf{v}}(\mathbf{w}_i(t), \mathbf{v}_i(t))\right\|^2\right].
\end{align}
\end{lemma}

\begin{proof}
We first rewrite $C_3$ as
\begin{align*}
    C_3&=\frac{1}{N}\sum_{i=1}^N\mathbb{E}\Big\langle \nabla_{\mathbf{v}}L_i(\bar{\mathbf{w}}(t),\mathbf{v}_i(t)),  -\eta_v \sum_{s=0}^{\tau-1} g_{\mathbf{v}}(\mathbf{w}_i(t),\mathbf{v}_i(t,s))\Big\rangle\allowdisplaybreaks\\
    &=\frac{1}{N}\sum_{i=1}^N\mathbb{E}\Big\langle \nabla_{\mathbf{v}}L_i(\bar{\mathbf{w}}(t),\mathbf{v}_i(t)),  \nonumber\\
    &\qquad\quad-\eta_v \sum_{s=0}^{\tau-1} g_{\mathbf{v}}(\mathbf{w}_i(t),\mathbf{v}_i(t,s))+\eta_v\tau \nabla_{\mathbf{v}}L_i(\bar{\mathbf{w}}(t),\mathbf{v}_i(t))\nonumber\\
    &\qquad\quad-\eta_v\tau \nabla_{\mathbf{v}}L_i(\bar{\mathbf{w}}(t),\mathbf{v}_i(t))\Big\rangle.
\end{align*}
Then, we have
\begin{align}\label{eq:C3}
    C_3=\quad&\frac{-\eta_v\tau}{N}\sum_{i=1}^N \mathbb{E}\left[\left\|\nabla_{\mathbf{v}}L_i(\bar{\mathbf{w}}(t),\mathbf{v}_i(t))\right\|^2\right]\nonumber\allowdisplaybreaks\\
    &\quad+\frac{1}{N}\sum_{i=1}^N\mathbb{E}\Big\langle \nabla_{\mathbf{v}}L_i(\bar{\mathbf{w}}(t),\mathbf{v}_i(t)),\nonumber\\
    &\quad-\eta_v \sum_{s=0}^{\tau-1} g_{\mathbf{v}}(\mathbf{w}_i(t),\mathbf{v}_i(t,s))+\eta_v\tau \nabla_{\mathbf{v}}L_i(\bar{\mathbf{w}}(t),\mathbf{v}_i(t))\Big\rangle\nonumber\allowdisplaybreaks\\
    =\quad&\frac{-\eta_v\tau}{N}\sum_{i=1}^N \mathbb{E}\left[\left\|\nabla_{\mathbf{v}}L_i(\bar{\mathbf{w}}(t),\mathbf{v}_i(t))\right\|^2\right]\nonumber\\
    &\quad+\frac{1}{N}\sum_{i=1}^N\mathbb{E}\Big\langle \sqrt{\eta_v\tau}\nabla_{\mathbf{v}}L_i(\bar{\mathbf{w}}(t),\mathbf{v}_i(t)), \nonumber\allowdisplaybreaks\\
    &\quad-\sqrt{\frac{\eta_v}{\tau}} \sum_{s=0}^{\tau-1} (g_{\mathbf{v}}(\mathbf{w}_i(t),\mathbf{v}_i(t,s))- \nabla_{\mathbf{v}}L_i(\bar{\mathbf{w}}(t),\mathbf{v}_i(t)))\Big\rangle\nonumber\allowdisplaybreaks\\
    \overset{(c_1)}{=}\quad&\frac{-\eta_v\tau}{N}\sum_{i=1}^N \mathbb{E}\Bigg[\Bigg\|\nabla_{\mathbf{v}}L_i(\bar{\mathbf{w}}(t),\mathbf{v}_i(t))Bigg\|^2\Bigg]\nonumber\\
    &\quad+\frac{1}{N}\sum_{i=1}^N\Bigg(\frac{\eta_v\tau}{2}\mathbb{E}\left[\left\|\nabla_{\mathbf{v}}L_i(\bar{\mathbf{w}}(t),\mathbf{v}_i(t))\right\|^2\right]\nonumber\allowdisplaybreaks\\
    &\quad-\frac{\eta_v}{2\tau}\mathbb{E}\left[\left\| \sum_{s=0}^{\tau-1} g_{\mathbf{v}}(\mathbf{w}_i(t),\mathbf{v}_i(t,s))\right\|^2\right]\nonumber\\
    &\quad+\frac{\eta_v}{2\tau}\mathbb{E}\Bigg[\Bigg\|\sum_{s=0}^{\tau-1} (g_{\mathbf{v}}(\mathbf{w}_i(t),\mathbf{v}_i(t,s))\nonumber\\
    &\quad-\nabla_{\mathbf{v}}L_i(\bar{\mathbf{w}}(t),\mathbf{v}_i(t)))\Bigg\|^2\Bigg]\Bigg)\nonumber\allowdisplaybreaks \\
    =\quad&\frac{-\eta_v\tau}{2N}\sum_{i=1}^N \mathbb{E}\left[\left\|\nabla_{\mathbf{v}}L_i(\bar{\mathbf{w}}(t),\mathbf{v}_i(t))\right\|^2\right]\nonumber\\
    &\quad-\frac{\eta_v}{2N\tau}\sum_{i=1}^N\mathbb{E}\left[\left\| \sum_{s=0}^{\tau-1} g_{\mathbf{v}}(\mathbf{w}_i(t),\mathbf{v}_i(t,s))\right\|^2\right]\nonumber\allowdisplaybreaks\\
    &\quad+\underset{D_2}{\underbrace{\frac{\eta_v}{2N\tau}\sum_{i=1}^N\mathbb{E}\Bigg[\Bigg\|\sum_{s=0}^{\tau-1} (g_{\mathbf{v}}(\mathbf{w}_i(t),\mathbf{v}_i(t,s))}}\nonumber\\
    &\quad\underset{D_2}{\underbrace{-\nabla_{\mathbf{v}}L_i(\bar{\mathbf{w}}(t),\mathbf{v}_i(t)))\Bigg\|^2\Bigg]}}\nonumber\allowdisplaybreaks \\
    \overset{(c_2)}{\leq}\quad& \frac{-\eta_v\tau}{2} \mathbb{E}\left[\left\|\nabla_{\mathbf{v}}L(\bar{\mathbf{w}}(t),\{\mathbf{v}_i(t))\}_{i=1}^N\right\|^2\right]\nonumber\\
    &\quad-\frac{\eta_v}{2N\tau}\sum_{i=1}^N\mathbb{E}\left[\left\| \sum_{s=0}^{\tau-1} g_{\mathbf{v}}(\mathbf{w}_i(t),\mathbf{v}_i(t,s))\right\|^2\right]\nonumber\allowdisplaybreaks\\
&\quad+\underset{D_2}{\underbrace{\frac{\eta_v}{2N\tau}\sum_{i=1}^N\mathbb{E}\Bigg[\Bigg\|\sum_{s=0}^{\tau-1} (g_{\mathbf{v}}(\mathbf{w}_i(t),\mathbf{v}_i(t,s))}}\nonumber\\
&\quad-\underset{D_2}{\underbrace{\nabla_{\mathbf{v}}L_i(\bar{\mathbf{w}}(t),\mathbf{v}_i(t)))\Bigg\|^2\Bigg]}},
\end{align}
where $(c_1)$ holds according to the equation $\langle a,b\rangle = \frac{1}{2}[\|a\|^2+\|b\|^2-\|a-b\|^2]$; and   $(c_2)$ follows the Cauchy-Schwartz inequality $\|\sum_{i=1}^N a_i\|^2\leq N\sum_{i=1}^N\|a_i\|^2$.
To this end, the key to bound $C_3$ is to bound $D_2$, which is given by
\begin{align}\label{eq:D2}
    D_2=\quad&\frac{\eta_v}{2N\tau}\sum_{i=1}^N\mathbb{E}\Bigg[\Bigg\|\sum_{s=0}^{\tau-1} g_{\mathbf{v}}(\mathbf{w}_i(t),\mathbf{v}_i(t,s))\nonumber\\
    &\quad-\sum_{s=0}^{\tau-1}\nabla_{\mathbf{v}}L_i(\bar{\mathbf{w}}(t),\mathbf{v}_i(t))\Bigg\|^2\Bigg]\nonumber\allowdisplaybreaks\\
    =\quad&\frac{\eta_v}{2N\tau}\sum_{i=1}^N\mathbb{E}\Bigg[\Bigg\|\sum_{s=0}^{\tau-1} g_{\mathbf{v}}(\mathbf{w}_i(t),\mathbf{v}_i(t,s))\nonumber\\
    &\quad-\nabla_{\mathbf{v}}L_i(\bar{\mathbf{w}}(t),\mathbf{v}_i(t,s))\nonumber\allowdisplaybreaks\\
&\quad+\nabla_{\mathbf{v}}L_i(\bar{\mathbf{w}}(t),\mathbf{v}_i(t,s))\nonumber\\
&\quad-\sum_{s=0}^{\tau-1}\nabla_{\mathbf{v}}L_i(\bar{\mathbf{w}}(t),\mathbf{v}_i(t))\Bigg\|^2\Bigg]\nonumber\allowdisplaybreaks\\
\leq\quad&\frac{\eta_v}{N\tau}\sum_{i=1}^N\mathbb{E}\Bigg[\Bigg\|\sum_{s=0}^{\tau-1} \nabla_{\mathbf{v}}L_i(\mathbf{w}_i(t),\mathbf{v}_i(t,s))\nonumber\\
&\quad- \nabla_{\mathbf{v}}L_i(\bar{\mathbf{w}}(t),\mathbf{v}_i(t,s))\Bigg\|^2\Bigg]\nonumber\allowdisplaybreaks\\
&\quad+\frac{\eta_v}{N\tau}\sum_{i=1}^N\mathbb{E}\Bigg[\Bigg\|\sum_{s=0}^{\tau-1}\nabla_{\mathbf{v}}L_i(\bar{\mathbf{w}}(t),\mathbf{v}_i(t,s))\nonumber\\
&\quad-\sum_{s=0}^{\tau-1}\nabla_{\mathbf{v}}L_i(\bar{\mathbf{w}}(t),\mathbf{v}_i(t))\Bigg\|^2\Bigg]\nonumber\allowdisplaybreaks\\
    \leq\quad& \frac{\eta_v\tau L^2}{N}\sum_{i=1}^N\mathbb{E}\left[\left\|\mathbf{w}_i(t)- \bar{\mathbf{w}}(t)\right\|^2\right]\nonumber\\
    &\quad+\frac{\eta_v L^2}{N}\sum_{i=1}^N\sum_{s=0}^{\tau-1}\mathbb{E}\left[\left\|\mathbf{v}_i(t,s)-\mathbf{v}_i(t)\right\|^2\right].
\end{align}
Next, we bound $\mathbb{E}\left[\left\|\mathbf{v}_i(t,s)-\mathbf{v}_i(t)\right\|^2\right]$ as follows.
\begin{align}\label{eq:D2_condition}
&\mathbb{E}\left[\|\mathbf{v}_i(t,s)-\mathbf{v}_{i}(t)\|^2\right]\nonumber\allowdisplaybreaks\\
=\quad&\mathbb{E}\Bigg[\|\mathbf{v}_i(t,s-1)-\mathbf{v}_{i}(t)-\eta_v g_{\mathbf{v}}(\mathbf{w}_i(t), \mathbf{v}_i(t,s-1))\|^2\Bigg]\nonumber\allowdisplaybreaks\\
=\quad&\mathbb{E}\Big[\|\mathbf{v}_i(t,s-1)-\mathbf{v}_{i}(t)-\eta_v g_{\mathbf{v}}(\mathbf{w}_i(t), \mathbf{v}_i(t,s-1))\nonumber\\
&\quad+\eta_v\nabla_{\mathbf{v}}L_i(\mathbf{w}_i(t), \mathbf{v}_i(t,s-1))\nonumber\\
&\quad-\eta_v\nabla_{\mathbf{v}}L_i(\mathbf{w}_i(t), \mathbf{v}_i(t,s-1))\nonumber\allowdisplaybreaks\\
&\quad+\eta_v\nabla_{\mathbf{v}}L_i(\mathbf{w}_i(t), \mathbf{v}_i(t)-\eta_v\nabla_{\mathbf{v}}L_i(\mathbf{w}_i(t), \mathbf{v}_i(t)))\nonumber\\
&\quad+\eta_v g_{\mathbf{v}}(\mathbf{w}_i(t), \mathbf{v}_i(t))-\eta_v g_{\mathbf{v}}(\mathbf{w}_i(t), \mathbf{v}_i(t))\|^2\Big]\nonumber\allowdisplaybreaks\\
\leq\quad& \left(1+\frac{1}{2\tau-1}\right)\mathbb{E}\left[\|\mathbf{v}_i(t,s-1)-\mathbf{v}_{i}(t)\|^2\right]\nonumber\\
&\quad+\eta_v^2\mathbb{E}\Bigg[\|g_{\mathbf{v}}(\mathbf{w}_i(t), \mathbf{v}_i(t,s-1))\nonumber\\
&\quad-\nabla_{\mathbf{v}}L_i(\mathbf{w}_i(t), \mathbf{v}_i(t,s-1))\|^2\Bigg]\nonumber\allowdisplaybreaks\\
&\quad+6\eta_v^2\tau \mathbb{E}\Bigg[\|\nabla_{\mathbf{v}}L_i(\mathbf{w}_i(t), \mathbf{v}_i(t,s-1)\nonumber\\
&\quad-\nabla_{\mathbf{v}}L_i(\mathbf{w}_i(t), \mathbf{v}_i(t)))\|^2\Bigg]\nonumber\allowdisplaybreaks\\
&\quad+6\eta_v^2\tau \mathbb{E}\left[\|\nabla_{\mathbf{v}}L_i(\mathbf{w}_i(t), \mathbf{v}_i(t)- g_{\mathbf{v}}(\mathbf{w}_i(t), \mathbf{v}_i(t)))\|^2\right]\nonumber\\
&\quad+6\eta_v^2\tau\mathbb{E}\left[\|g_{\mathbf{v}}(\mathbf{w}_i(t), \mathbf{v}_i(t))\|^2\right]\nonumber\allowdisplaybreaks\\
\leq\quad& \Bigg(1+\frac{1}{2\tau-1}\Bigg)\mathbb{E}\left[\|\mathbf{v}_i(t,s-1)-\mathbf{v}_{i}(t)\|^2\right]+\eta_v^2\sigma^2\nonumber\\
&\quad+6\eta_v^2\tau L^2\mathbb{E}\left[\|\mathbf{v}_i(t,s-1)-\mathbf{v}_{i}(t)\|^2\right]\nonumber\allowdisplaybreaks\\
&\quad+6\eta_v^2\tau \sigma^2+6\eta_v^2\tau\mathbb{E}\left[\|g_{\mathbf{v}}(\mathbf{w}_i(t), \mathbf{v}_i(t))\|^2\right]\nonumber\allowdisplaybreaks\\    
\leq\quad& \left(1+\frac{1}{2\tau-1}+6\eta_v^2\tau L^2\right)\mathbb{E}\left[\|\mathbf{v}_i(t,s-1)-\mathbf{v}_{i}(t)\|^2\right]\nonumber\\
&\quad+(\eta_v^2+6\eta_v^2\tau)\sigma^2+6\eta_v^2\tau\mathbb{E}\|g_{\mathbf{v}}(\mathbf{w}_i(t), \mathbf{v}_i(t))\|^2 \nonumber\allowdisplaybreaks\\
\leq\quad& (18\tau^2-15\tau-3)\eta_v^2\sigma^2\nonumber\\
&\quad+(18\tau^2-18\tau)\eta_v^2\mathbb{E}\left[\|g_{\mathbf{v}}(\mathbf{w}_i(t), \mathbf{v}_i(t))\|^2\right],
\end{align}
where the last inequality holds by using Lemma 3 in \cite{reddi2020adaptive}.
Substituting \eqref{eq:D2_condition} into \eqref{eq:D2}, we have
\begin{align}\label{eq:D2_2}
D_2\leq\quad& \frac{\eta_v\tau L^2}{N}\sum_{i=1}^N\mathbb{E}\left[\left\|\mathbf{w}_i(t)- \bar{\mathbf{w}}(t)\right\|^2\right]\nonumber\\
&\quad+\eta_v^3L^2(18\tau^3-15\tau^2-3\tau)\sigma^2\nonumber\allowdisplaybreaks\\
&\quad+\frac{\eta_v^3L^2(18\tau^3-18\tau^2)}{N}\sum_{i=1}^N\mathbb{E}\left[\left\|g_{\mathbf{v}}(\mathbf{w}_i(t), \mathbf{v}_i(t))\right\|^2\right].
\end{align}
Hence, substituting $D_2$ in \eqref{eq:D2_2} back to \eqref{eq:C3}, we have the desired result as follows
\begin{align*}
    C_3{\leq}\quad &\frac{-\eta_v\tau}{2} \mathbb{E}\left[\left\|\nabla_{\mathbf{v}}L(\bar{\mathbf{w}}(t),\{\mathbf{v}_i(t))\}_{i=1}^N\right\|^2\right]\nonumber\\
    &\quad-\frac{\eta_v}{2N\tau}\sum_{i=1}^N\mathbb{E}\left[\left\| \sum_{s=0}^{\tau-1} g_{\mathbf{v}}(\mathbf{w}_i(t),\mathbf{v}_i(t,s))\right\|^2\right]\nonumber\allowdisplaybreaks\\
&\quad+\frac{\eta_v\tau L^2}{N}\sum_{i=1}^N\mathbb{E}\left[\left\|\mathbf{w}_i(t)- \bar{\mathbf{w}}(t)\right\|^2\right]\nonumber\\
&\quad+\eta_v^3L^2(18\tau^3-15\tau^2-3\tau)\sigma^2\nonumber\\
&\quad+\frac{\eta_v^3L^2(18\tau^3-18\tau^2)}{N}\sum_{i=1}^N\mathbb{E}\left[\left\|g_{\mathbf{v}}(\mathbf{w}_i(t), \mathbf{v}_i(t))\right\|^2\right].
\end{align*}
\end{proof}

\begin{lemma}\label{lem:C4}
$C_4$ can be bounded as
\begin{align}
C_4{\leq}\frac{\eta_v^2\tau L}{2}\sigma_2^2+\frac{\eta_v^2L}{2N}\sum_{i=1}^N\left\|\mathbb{E}\sum_{s=0}^{\tau-1}g_{\mathbf{v}}(\mathbf{w}_i(t),\mathbf{v}_i(t,s))\right\|^2. 
\end{align}
\end{lemma}
\begin{proof}
The proof directly follows from the definition and the proof for $C_2$ in Lemma \ref{lem:C2}, i.e.,
\begin{align}
   C_4=\quad&\frac{L}{2N}\sum_{i=1}^N\mathbb{E}[\|{\mathbf{v}}_i(t+1)-{\mathbf{v}}_i(t)\|^2]\nonumber\displaybreak[0]\\
   =\quad&\frac{\eta_v^2L}{2N}\sum_{i=1}^N\mathbb{E}\left[\left\| \sum_{s=0}^{\tau-1} g_{\mathbf{v}}(\mathbf{w}_i(t),\mathbf{v}_i(t,s))\right\|^2\right]\nonumber\displaybreak[1]\\
{=}\quad&\frac{\eta_v^2L}{2N}\sum_{i=1}^N\mathbb{E}\Bigg[\Bigg\|\sum_{s=0}^{\tau-1} g_{\mathbf{v}}(\mathbf{w}_i(t), \mathbf{v}_i(t,s))\nonumber\\
&\quad-\nabla_{\mathbf{v}}L_i(\mathbf{w}_i(t),\mathbf{v}_i(t,s))\Bigg\|^2\Bigg]\nonumber\\
&\quad+\frac{\eta_v^2L}{2N}\sum_{i=1}^N\left\|\mathbb{E} \sum_{s=0}^{\tau-1} g_{\mathbf{v}}(\mathbf{w}_i(t), \mathbf{v}_i(t,s))\right\|^2 \nonumber\displaybreak[2]\\
{\leq}\quad&\frac{\eta_v^2\tau L}{2}\sigma_2^2+\frac{\eta_v^2L}{2N}\sum_{i=1}^N\left\|\mathbb{E} \sum_{s=0}^{\tau-1} g_{\mathbf{v}}(\mathbf{w}_i(t), \mathbf{v}_i(t,s))\right\|^2. 
\end{align}
\end{proof}


Now we are ready to prove the main results.  Given the bounds in Lemmas~\ref{lem:C1},~\ref{lem:C2},~\ref{lem:C3} and~\ref{lem:C4} for $C_1, C_2, C_3$ and $C_4$, respectively, we substitute these bound back to \eqref{eq:global_iter2} and obtain
\begin{align}\label{eq:global_iter3}
 \mathbb{E}&[L(\bar{\mathbf{w}}(t+1),\{\mathbf{v}_i(t+1)\}_{i=1}^N)]- \mathbb{E}[L(\bar{\mathbf{w}}(t),\{\mathbf{v}_i(t)\}_{i=1}^N)]\nonumber\displaybreak[0]\\
 {\leq}\quad&{\frac{\eta_w L^2}{2N}\sum_{i=1}^N\mathbb{E}\left[\left\|\bar{\mathbf{w}}(t)-\mathbf{w}_i(t)\right\|^2\right]}\nonumber\\
 &\quad-\frac{\eta_w}{2}\left\|\nabla_{\mathbf{w}}L(\bar{\mathbf{w}}(t),\{\mathbf{v}_i(t+1)\}_{i=1}^N)\right\|^2\nonumber\allowdisplaybreaks\\
 &\quad-\frac{\eta_w}{2N^2}\left\|\mathbb{E}\sum_{i=1}^N g_{\mathbf{w}}(\mathbf{w}_i(t),\mathbf{v}_i(t+1))\right\|^2+\frac{\eta_w^2L}{2N}\sigma_1^2\nonumber\\
 &\quad+\frac{\eta_w^2L}{2N^2} \left\|\mathbb{E}\sum_{i=1}^N g_{\mathbf{w}}(\mathbf{w}_i(t), \mathbf{v}_i(t+1))\right\|^2\nonumber \\ \displaybreak[3]
&\quad-\frac{\eta_v\tau}{2} \left\|\nabla_{\mathbf{v}}L(\bar{\mathbf{w}}(t),\{\mathbf{v}_i(t))\}_{i=1}^N\right\|^2\nonumber\\
&\quad-\frac{\eta_v}{2N\tau}\sum_{i=1}^N\left\| \mathbb{E}\sum_{s=0}^{\tau-1} g_{\mathbf{v}}(\mathbf{w}_i(t),\mathbf{v}_i(t,s))\right\|^2\nonumber\allowdisplaybreaks\\
&\quad+\frac{\eta_v\tau L^2}{N}\sum_{i=1}^N\mathbb{E}\left[\left\|\mathbf{w}_i(t)- \bar{\mathbf{w}}(t)\right\|^2\right]\nonumber\\
&\quad+\eta_v^3L^2(18\tau^3-15\tau^2-3\tau)\sigma_2^2\nonumber\allowdisplaybreaks\\
&\quad+\frac{\eta_v^3L^2(18\tau^3-18\tau^2)}{N}\sum_{i=1}^N\mathbb{E}\left[\left\|g_{\mathbf{v}}(\mathbf{w}_i(t), \mathbf{v}_i(t))\right\|^2\right]\nonumber\\
&\quad+\frac{\eta_v^2L}{2N}\sum_{i=1}^N\left\|\mathbb{E}\sum_{s=0}^{\tau-1}g_{\mathbf{v}}(\mathbf{w}_i(t),\mathbf{v}_i(t,s))\right\|^2+\frac{\eta_v^2\tau L}{2}\sigma_2^2\nonumber \allowdisplaybreaks\\
&={\frac{(2\eta_v\tau+\eta_w) L^2}{2N}\sum_{i=1}^N\mathbb{E}\left[\left\|\bar{\mathbf{w}}(t)-\mathbf{w}_i(t)\right\|^2\right]}+\frac{\eta_w^2L}{2N}\sigma_1^2\nonumber\\
&\quad+\eta_v^3L^2(18\tau^3-15\tau^2-3\tau)\sigma_2^2+\frac{\eta_v^2\tau L}{2}\sigma_2^2\nonumber\allowdisplaybreaks\\
&\quad+\frac{\eta_w^2L-\eta_w}{2N^2}\left\|\mathbb{E}\sum_{i=1}^N g_{\mathbf{w}}(\mathbf{w}_i(t),\mathbf{v}_i(t+1))\right\|^2\nonumber\\
&\quad-\frac{\eta_w}{2}\left\|\nabla_{\mathbf{w}}L(\bar{\mathbf{w}}(t),\{\mathbf{v}_i(t+1)\}_{i=1}^N)\right\|^2\nonumber\allowdisplaybreaks\\
&\quad +\frac{\eta_v^2\tau L-\eta_v}{2N\tau}\sum_{i=1}^N\left\|\mathbb{E}\sum_{s=0}^{\tau-1} g_{\mathbf{v}}(\mathbf{w}_i(t),\mathbf{v}_i(t,s))\right\|^2\nonumber\\
&\quad-\frac{\eta_v\tau}{2}\left\|\nabla_{\mathbf{v}}L(\bar{\mathbf{w}}(t),\{\mathbf{v}_i(t)\}_{i=1}^N)\right\|^2\nonumber\allowdisplaybreaks\\
&\quad +\frac{18\eta_v^3 L^2(\tau^3-\tau^2)}{N}\sum_{i=1}^N\mathbb{E}\left[\left\|g_{\mathbf{v}}(\mathbf{w}_i(t), \mathbf{v}_i(t))\right\|^2\right],
\end{align}
where the first inequality comes from the fact that $\|\mathbb{E}[X]\|^2\leq \mathbb{E}[\|X\|^2]$.
To characterize the convergence rate, the key then boils down to bound  the consensus error of the global representation $\mathbf{w}$, i.e., $\mathbb{E}\left[\|\mathbf{w}_i(t)-\bar{\mathbf{w}}(t)\|^2\right], \forall i$. We bound it in the following lemma.

We present the following lemma from \cite{nedic2009distributed} for completeness. Then we continued to use Lemma~\ref{lemma:matrix} to certify Lemma~\ref{lemma:consensus_error}.
\begin{lemma}[Lemma 4 in \cite{nedic2009distributed}]\label{lemma:matrix}
Assume that $\mathbf{P}$ is doubly stochastic.  The difference between $1/N$ and any element of   $\mathbf{P}^{k-s}$ can be bounded by
\begin{align}
\left|\frac{1}{N}-\mathbf{P}^{k-s}(i,j)\right|\leq 2 \frac{(1+p^{-N})}{1-p^{N}}(1-p^{N})^{(t-s)/N},
\end{align}
where  $p$ is the smallest positive value of all consensus matrices, i.e., $p=\arg\min \mathbf{P}_{i,j}$ with $\mathbf{P}_{i,j}>0, \forall i,j.$
\label{lemma_bound_Phi}
\end{lemma}

\begin{lemma}\label{lemma:consensus_error}
The consensus error $\mathbb{E}\left[\|\mathbf{w}_i(t)-\bar{\mathbf{w}}(t)\|^2\right], \forall i, k\geq 1$ is upper bounded by
\begin{align}
&\mathbb{E}\left[\|\mathbf{w}_i(t)-\bar{\mathbf{w}}(t)\|^2\right]\nonumber\\
\!\leq\quad& \!\frac{6\eta_w^2C^2N^2}{(1-q)^2}\sigma_1^2\!+\!\frac{18\eta_w^2C^2N^2}{(1-q)^2}\varsigma^2\nonumber\\
&\quad+\frac{18\eta_w^2C^2L^2N}{1-q}\sum_{r=0}^{k-1}q^{k-r}\sum_{j=1}^N\mathbb{E}\|\mathbf{w}_j(r)-\bar{\mathbf{w}}(r)\|^2\nonumber\displaybreak[0]\\
&\quad+\frac{18\eta_w^2C^2N^2}{1-q}\sum_{r=0}^{k-1}q^{k-r}\mathbb{E}\|\nabla_{\mathbf{w}} L(\bar{\mathbf{w}}(r),\{\mathbf{v}_j(r+1)\}_{j=1}^N)\|^2.
\end{align}

\end{lemma}
\begin{proof}
Based on the update method and communication rule of shared adapter $w_i(t)$,  we have $\forall k\geq 1$
\begin{align}
&\mathbb{E}\left[\|\mathbf{w}_i(t)-\bar{\mathbf{w}}(t)\|^2 \right]\nonumber\displaybreak[0]\\ 
=\quad&\mathbb{E}\Bigg\|\frac{1}{N}\sum\limits_{j=1}^{N} \mathbf{w}_j(0)-\sum\limits_{j=1}^{N}\mathbf{w}_j(0)P^{k}(i,j)\nonumber\\
&\quad-\frac{\eta_w}{N}\sum\limits_{r=0}^{k-1}\sum\limits_{j=1}^{N}g_{\mathbf{w}}(\mathbf{w}_j(r), \mathbf{v}_j(r+1))
\displaybreak[0]\nonumber\displaybreak[1]\\ 
&\quad+\eta_w\sum\limits_{r=0}^{k-1}\sum\limits_{j=1}^{N}g_{\mathbf{w}}(\mathbf{w}_j(r),\mathbf{v}_j(r+1))P^{k-r}(i,j) \Bigg\|^2\nonumber\allowdisplaybreaks\\
=\quad&\mathbb{E}\Bigg\|\sum\limits_{j=1}^{N} \mathbf{w}_j(0)\Bigg(\frac{1}{N}-P^k(i,j)\Bigg)\nonumber\\
&\quad-\eta_w\sum\limits_{r=0}^{k-1}\sum\limits_{j=1}^{N}g_{\mathbf{w}}(\mathbf{w}_j(r),\mathbf{v}_j(r+1))\nonumber\\
&\quad\cdot\left(\frac{1}{N}-P^{k-r}(i,j)\right) \Bigg\|^2\displaybreak[0]\nonumber\\
\overset{(e_1)}{\leq}\quad&\mathbb{E}\left\|\sum\limits_{j=1}^{N} \mathbf{w}_j(0)\left(\frac{1}{N}-P^k(i,j)\right)\right\|^2\nonumber\\
&\quad+\mathbb{E}\Bigg\|\eta_w\sum\limits_{r=0}^{k-1}\sum\limits_{j=1}^{N}g_{\mathbf{w}}(\mathbf{w}_j(r),\mathbf{v}_j(r+1))\nonumber\\
&\qquad\quad\cdot\left(\frac{1}{N}-P^{k-r}(i,j)\right)\Bigg\|^2\displaybreak[0] \nonumber\\
\overset{(e_2)}{\leq}\quad&\mathbb{E}\Bigg\|\sum\limits_{j=1}^{N}2\mathbf{w}_j(0)\frac{1+p^{-N}}{1-p^{N}}(1-p^{N})^{(t-1)/N}\Bigg\|^2\displaybreak[0] \nonumber\\
&\quad+\mathbb{E}\Bigg\|\eta_w\sum\limits_{r=0}^{k-1}\sum\limits_{j=1}^{N}2g_{\mathbf{w}}(\mathbf{w}_j(r),\mathbf{v}_j(r+1))\nonumber\\
&\qquad\quad\cdot\frac{1+p^{-N}}{1-p^{N}}(1-p^{N})^{(t-r)/N} \Bigg\|^2\displaybreak[0]\nonumber\\
\overset{(e_3)}{=}\quad&\mathbb{E}\Bigg\|\eta_w C\sum\limits_{r=0}^{k-1}q^{(t-r)}\sum\limits_{j=1}^{N}g_{\mathbf{w}}(\mathbf{w}_j(r),\mathbf{v}_j(r+1)) \Bigg\|^2\nonumber \displaybreak[0]\\
=\quad&\mathbb{E}\Bigg\|\eta_w C\sum\limits_{r=0}^{k-1}q^{(t-r)}\sum\limits_{j=1}^{N}(g_{\mathbf{w}}(\mathbf{w}_j(r),\mathbf{v}_j(r+1))\nonumber\\
&\quad-\nabla_{\mathbf{w}}F_j(\mathbf{w}_j(r),\mathbf{v}_j(r+1)))\nonumber\displaybreak[0]\\
&\quad+\eta_w C\sum\limits_{r=0}^{k-1}q^{(t-r)}\sum\limits_{j=1}^{N}\nabla_{\mathbf{w}}F_j(\mathbf{w}_j(r),\mathbf{v}_j(r+1)) \Bigg\|^2\displaybreak[0]\nonumber\\
\overset{(e_4)}{\leq}\quad&\mathbb{E}\Bigg\|\eta_w C\sum\limits_{r=0}^{k-1}q^{(t-r)}\sum\limits_{j=1}^{N}(g_{\mathbf{w}}(\mathbf{w}_j(r),\mathbf{v}_j(r+1))\nonumber\\
&\quad-\nabla_{\mathbf{w}}F_j(\mathbf{w}_j(r),\mathbf{v}_j(r+1)))\Bigg\|^2\displaybreak[0]\nonumber\\
&\quad +\mathbb{E}\|2\eta_w^2 C^2\sum\limits_{r=0}^{k-1}\sum\limits_{r^\prime=0}^{k-1}q^{(2k-r-r^\prime)}\|\nonumber\\
&\cdot\Bigg\|\sum\limits_{j=1}^{N}(g_{\mathbf{w}}(\mathbf{w}_j(r),\mathbf{v}_j(r+1))-\nabla_{\mathbf{w}}F_j(\mathbf{w}_j(r),\mathbf{v}_j(r+1)))\Bigg\|\displaybreak[0]\nonumber\\
&\cdot\Bigg\|\sum\limits_{j=1}^{N}\nabla_{\mathbf{w}}F_j(\mathbf{w}_j(r),\mathbf{v}_j(r+1))\Bigg\|\nonumber\\
&\quad+\mathbb{E}\Bigg\|\eta_w C\sum\limits_{r=0}^{k-1}q^{(t-r)}\sum\limits_{j=1}^{N}\nabla_{\mathbf{w}}F_j(\mathbf{w}_j(r),\mathbf{v}_j(r+1))\Bigg\|^2\displaybreak[0]\nonumber\\
\overset{(e_5)}{\leq} \quad&\mathbb{E}\frac{6\eta_w^2C^2}{1-q}\sum_{r=0}^{k-1} q^{k-r}\nonumber\\
&\cdot\Bigg\|\sum\limits_{j=1}^{N}(g_{\mathbf{w}}(\mathbf{w}_j(r),\mathbf{v}_j(r+1))-\nabla_{\mathbf{w}}F_j(\mathbf{w}_j(r),\mathbf{v}_j(r+1)))\Bigg\|^2\nonumber\allowdisplaybreaks\\
&\quad+ \mathbb{E}\frac{6\eta_w^2C^2}{1-q}\sum_{r=0}^{k-1} q^{k-r}\Bigg\|\sum\limits_{j=1}^{N}\nabla_{\mathbf{w}}F_j(\mathbf{w}_j(r),\mathbf{v}_j(r+1))\Bigg\|^2\displaybreak[0]\nonumber\\
\overset{(e_6)}{\leq} \quad&\frac{6\eta_w^2C^2N^2}{(1-q)^2}\sigma_1^2+\frac{18\eta_w^2C^2N^2}{(1-q)^2}\varsigma^2\nonumber\\
&\quad+\frac{18\eta_w^2C^2L^2N}{1-q}\sum_{r=0}^{k-1}q^{k-r}\sum_{j=1}^N\mathbb{E}\|\mathbf{w}_j(r)-\bar{\mathbf{w}}(r)\|^2\nonumber\displaybreak[0]\\
&\quad+\frac{18\eta_w^2C^2N^2}{1-q}\sum_{r=0}^{k-1}q^{k-r}\mathbb{E}\|\nabla_{\mathbf{w}} L(\bar{\mathbf{w}}(r),\{\mathbf{v}_j(r+1)\}_{j=1}^N)\|^2,
\end{align}
where $(e_1)$  is due to the inequality $\|a-b\|^2\leq 2\|a\|^2+2\|b\|^2$; $(e_2)$ holds according to Lemma \ref{lemma_bound_Phi}. W.l.o.g., we assume that the initial term $\mathbf{w}_i(0), \forall i$ is small enough and can be neglected;   
$(e_3)$ follows $C:=2\sqrt{2}\cdot\frac{1+p^{-N}}{1-p^{N}}$ and $q:=(1-p^{N})^{1/N}$;
$(e_4)$ is due to $\|a+b\|^2\leq \|a\|^2+\|b\|^2+2ab$; and $(e_5)$ is the standard mathematical manipulation by leveraging the following inequality, i.e., for any $q\in(0,1)$ and non-negative sequence $\{\chi(r)\}_{r=0}^{k-1}$, it holds \citep{assran2019stochastic}
    $\sum_{k=1}^{K-1}\sum_{r=0}^{k-1} q^{k-r}\chi(r)\leq \frac{1}{1-q}\sum_{r=0}^{K-1}\chi(r)$;
and $(e_6)$ holds due to Assumption \ref{assumption-gradient}, the inequality
\begin{align*}
    \sum_{k=1}^{K-1} q^k\sum_{r=0}^{k-1} q^{k-r}\chi(r)\leq \sum_{k=1}^{K-1}\sum_{r=0}^{k-1} q^{2(t-r)}\chi(r)\leq \frac{1}{1-q^2}\sum_{r=0}^{K-1}\chi(r),
\end{align*}
and the fact that
\begin{align*}
    &\frac{1}{N}\sum_{i=1}^N\mathbb{E}\|\nabla_{\mathbf{w}} L_i(\mathbf{w}_i(t-1), \mathbf{v}_i(t))\|^2\allowdisplaybreaks\\
    \leq\quad& \frac{1}{N}\sum_{i=1}^N\mathbb{E}\|\nabla_{\mathbf{w}} L_i(\mathbf{w}_i(t-1), \mathbf{v}_i(t))\nonumber\\
    &\quad-\nabla_{\mathbf{w}} L_i(\bar{\mathbf{w}}(t-1), \mathbf{v}_i(t))+\nabla_{\mathbf{w}} L_i(\bar{\mathbf{w}}(t-1), \mathbf{v}_i(t))\allowdisplaybreaks\\
    &\quad-\nabla_{\mathbf{w}} L(\bar{\mathbf{w}}(t-1), \{\mathbf{v}_i(t)\}_{i=1}^N)+\nabla_{\mathbf{w}} L(\bar{\mathbf{w}}(t-1), \{\mathbf{v}_i(t)\}_{i=1}^N)\|^2\allowdisplaybreaks\\
    &\leq \underset{\text{Lipschitz continuous gradient in Assumption \ref{assumption-lipschitz}}}{\underbrace{\frac{3}{N}\sum_{i=1}^N\mathbb{E}\|\nabla_{\mathbf{w}} L_i(\mathbf{w}_i(t-1), \mathbf{v}_i(t))-\nabla_{\mathbf{w}} L_i(\bar{\mathbf{w}}(t-1)\|^2}}\allowdisplaybreaks\\
   &\quad+\underset{\text{Bounded global variability in Assumption \ref{assumption:global-var}}}{\underbrace{\frac{3}{N}\sum_{i=1}^N\mathbb{E}\|\nabla_{\mathbf{w}} L_i(\bar{\mathbf{w}}(t-1), \mathbf{v}_i(t))}}\nonumber\nonumber\\
   &\underset{\text{Bounded global variability in Assumption \ref{assumption:global-var}}}{\underbrace{\quad-\nabla_{\mathbf{w}} L(\bar{\mathbf{w}}(t-1), \{\mathbf{v}_i(t)\}_{i=1}^N)\|^2}}\allowdisplaybreaks\\
   &\quad+\frac{3}{N}\sum_{i=1}^N\mathbb{E}\|\nabla_{\mathbf{w}} L(\bar{\mathbf{w}}(t-1), \{\mathbf{v}_i(t)\}_{i=1}^N)\|^2\allowdisplaybreaks\\
    \leq\quad& \frac{3L^2}{N}\sum_{i=1}^N\mathbb{E}\|\mathbf{w}_i(t)-\bar{\mathbf{w}}(t)\|^2+3\varsigma^2\nonumber\\
    &\quad+3\mathbb{E}\|\nabla_{\mathbf{w}} L(\bar{\mathbf{w}}(t-1), \{\mathbf{v}_i(t)\}_{i=1}^N)\|^2.
\end{align*}
This completes the proof.
\end{proof}

Rearrange the order of each term in \eqref{eq:global_iter3} and let $\max(L\eta_w, \eta_v\tau L(1+36\tau^2))\leq 1$, we have
\begin{align} \label{eq:gradient_2}
&\frac{\eta_v\tau}{2}\left\|\nabla_{\mathbf{v}}L(\bar{\mathbf{w}}(t),\{\mathbf{v}_i(t)\}_{i=1}^N)\right\|^2\nonumber\\
&\quad+\frac{\eta_w}{2}\left\|\nabla_{\mathbf{w}}L(\bar{\mathbf{w}}(t),\{\mathbf{v}_i(t+1)\}_{i=1}^N)\right\|^2\nonumber\\ \displaybreak[0]
\leq \quad&\mathbb{E}[L(\bar{\mathbf{w}}(t),\{\mathbf{v}_i(t)\}_{i=1}^N)]-\mathbb{E} [L(\bar{\mathbf{w}}(t+1),\{\mathbf{v}_i(t+1)\}_{i=1}^N)] \nonumber\\ \displaybreak[0]
&\quad+ {\frac{(2\eta_v\tau+\eta_w) L^2}{2N}\sum_{i=1}^N\mathbb{E}\left\|\bar{\mathbf{w}}(t)-\mathbf{w}_i(t)\right\|^2}+\frac{\eta_w^2L}{2N}\sigma_1^2\nonumber\\
&\quad+\eta_v^3L^2(18\tau^3-15\tau^2-3\tau)\sigma_2^2+\frac{\eta_v^2\tau L}{2}\sigma_2^2\nonumber\\
&\quad+\underset{(1):(1)+(2)\leq 0 }{\underbrace{\frac{\eta_v^2\tau L-\eta_v}{2N\tau}\sum_{i=1}^N\mathbb{E}\left\|\sum_{s=0}^{\tau-1} g_{\mathbf{v}}(\mathbf{w}_i(t),\mathbf{v}_i(t,s))\right\|^2}}\nonumber\\
&\underset{(2):(1)+(2)\leq 0 }{\underbrace{\quad+\frac{18\eta_v^3L^2(\tau^3-\tau^2)}{N}\sum_{i=1}^N\mathbb{E}\left\|g_{\mathbf{v}}(\mathbf{w}_i(t), \mathbf{v}_i(t))\right\|^2}}\nonumber\allowdisplaybreaks\\
&\quad +\underset{\leq 0 }{\underbrace{\frac{\eta_w^2L-\eta_w}{2N^2}\mathbb{E}\left\|\sum_{i=1}^N g_{\mathbf{w}}(\mathbf{w}_i(t),\mathbf{v}_i(t+1))\right\|^2}}\nonumber\allowdisplaybreaks\\
\leq\quad& \mathbb{E}[L(\bar{\mathbf{w}}(t),\{\mathbf{v}_i(t)\}_{i=1}^N)]-\mathbb{E} [L(\bar{\mathbf{w}}(t+1),\{\mathbf{v}_i(t+1)\}_{i=1}^N)] \nonumber\\ \displaybreak[0]
&\quad+ {\frac{(2\eta_v\tau+\eta_w) L^2}{2N}\sum_{i=1}^N\mathbb{E}\left\|\bar{\mathbf{w}}(t)-\mathbf{w}_i(t)\right\|^2}+\frac{\eta_w^2L}{2N}\sigma_1^2\nonumber\\
&\quad+\eta_v^3L^2(18\tau^3-15\tau^2-3\tau)\sigma_2^2+\frac{\eta_v^2\tau L}{2}\sigma_2^2.
\end{align}

According to Lemma \ref{lemma:consensus_error}, we have the following inequality
\begin{align}
&\sum_{k=0}^{K-1}\sum\limits_{i=1}^{N}\mathbb{E}\|\mathbf{w}_i(t)-\bar{\mathbf{w}}(t)\|^2  \nonumber\allowdisplaybreaks\\
\leq\quad &\sum\limits_{i=1}^{N}\mathbb{E}\|\mathbf{w}_i(0)-\bar{\mathbf{w}}(0)\|^2\nonumber\\
&\quad+\sum_{k=1}^{K-1}\Bigg(\frac{18N^2\eta_w^2C^2L^2}{1-q}\sum_{r=0}^{k-1}q^{k-r}\sum_{i=1}^N\mathbb{E}\|\mathbf{w}_i(r)-\bar{\mathbf{w}}(r)\|^2\nonumber\allowdisplaybreaks\\
&\quad+\frac{18\eta_w^2C^2N^3}{1-q}\sum_{r=0}^{k-1}q^{k-r}\mathbb{E}\|\nabla_{\mathbf{w}} L(\bar{\mathbf{w}}(r),\{\mathbf{v}_i(r+1)\}_{i=1}^N)\|^2\nonumber\\
&\quad+\frac{6\eta_w^2C^2N^3}{(1-q)^2}\sigma_1^2+\frac{18\eta_w^2C^2N^3}{(1-q)^2}\varsigma^2\Bigg)\displaybreak[0]\nonumber\allowdisplaybreaks\\
\leq\quad& \Bigg(\frac{18N^2\eta_w^2C^2L^2}{(1-q)^2}\sum_{k=0}^{K-1}\sum_{i=1}^N\mathbb{E}\|\mathbf{w}_i(t)-\bar{\mathbf{w}}(t)\|^2 \nonumber\allowdisplaybreaks\\
&\quad+\frac{18\eta_w^2C^2N^3}{(1-q)^2}\sum_{k=0}^{K-1}\mathbb{E}\|\nabla_{\mathbf{w}} L(\bar{\mathbf{w}}(t),\{\mathbf{v}_i(t+1)\}_{i=1}^N)\|^2\nonumber\\
&\quad+\frac{6K\eta_w^2C^2N^3}{(1-q)^2}\sigma_1^2+\frac{18K\eta_w^2C^2N^3}{(1-q)^2}\varsigma^2\Bigg)\displaybreak[0]\nonumber\allowdisplaybreaks\\
\leq\quad&\frac{18\eta_w^2C^2N^3}{(1-q)^2-18N^2\eta_w^2C^2L^2}\sum_{k=0}^{K-1}\mathbb{E}\|\nabla_{\mathbf{w}} L(\bar{\mathbf{w}}(t),\{\mathbf{v}_i(t+1)\}_{i=1}^N)\|^2\nonumber\\
&\quad+\frac{6K\eta_w^2C^2N^3}{(1-q)^2-18N^2\eta_w^2C^2L^2}\sigma_1^2\nonumber\allowdisplaybreaks\\
&\quad+\frac{18K\eta_w^2C^2N^3}{(1-q)^2-18N^2\eta_w^2C^2L^2}\varsigma^2,
\end{align}
where the second inequality holds due to the initialization such that $\sum_{i=1}^N\mathbb{E}\|\mathbf{w}_i(0)-\bar{\mathbf{w}}(t)\|^2=0.$ 
Summing the recursion in \eqref{eq:gradient_2} from round $0$ to round $K-1$  yields
\begin{align}
\sum\limits_{k=0}^{K-1}& \frac{\eta_v\tau}{2}\left\|\nabla_{\mathbf{v}}L(\bar{\mathbf{w}}(t),\{\mathbf{v}_i(t)\}_{i=1}^N)\right\|^2\nonumber\\
&\quad+\frac{\eta_w}{2}\left\|\nabla_{\mathbf{w}}L(\bar{\mathbf{w}}(t),\{\mathbf{v}_i(t+1)\}_{i=1}^N)\right\|^2\nonumber\\
\leq\quad &\mathbb{E}[L(\bar{\mathbf{w}}(0),\{\mathbf{v}_i(0)\}_{i=1}^N)]-\mathbb{E} [L(\bar{\mathbf{w}}(t),\{\mathbf{v}_i(t)\}_{i=1}^N)]\nonumber\\
&\quad + {\frac{(2\eta_v\tau+\eta_w) L^2}{2N}\sum_{k=0}^{K-1}\sum_{i=1}^N\mathbb{E}\left\|\bar{\mathbf{w}}(t)-\mathbf{w}_i(t)\right\|^2}+\frac{K\eta_w^2L}{2N}\sigma_1^2\nonumber\allowdisplaybreaks\\
&\quad+K\eta_v^3L^2(18\tau^3-15\tau^2-3\tau)\sigma_2^2+\frac{K\eta_v^2\tau L}{2}\sigma_2^2\nonumber\\
\leq\quad &\mathbb{E}[L(\bar{\mathbf{w}}(0),\{\mathbf{v}_i(0)\}_{i=1}^N)]-\mathbb{E} [L(\bar{\mathbf{w}}(t),\{\mathbf{v}_i(t)\}_{i=1}^N)]\nonumber\allowdisplaybreaks\\
&\quad+\frac{K\eta_w^2L}{2N}\sigma_1^2+K\eta_v^3L^2(18\tau^3-15\tau^2-3\tau)\sigma_2^2+\frac{K\eta_v^2\tau L}{2}\sigma_2^2\nonumber\displaybreak[0]\\
&\quad+\frac{9\eta_w^2C^2N^2L^2(2\eta_v\tau+\eta_w)}{(1-q)^2-18N^2\eta_w^2C^2L^2}\nonumber\\
&\quad\cdot\sum_{k=0}^{K-1}\mathbb{E}\|\nabla_{\mathbf{w}} L(\bar{\mathbf{w}}(t),\{\mathbf{v}_i(t)\}_{i=1}^N)\|^2\nonumber\displaybreak[0]\\
&\quad+\frac{3K\eta_w^2C^2N^2L^2(2\eta_v\tau+\eta_w)}{(1-q)^2-18N^2\eta_w^2C^2L^2}\sigma_1^2\nonumber\\
&\quad+\frac{9K\eta_w^2C^2N^2L^2(2\eta_v\tau+\eta_w)}{(1-q)^2-18N^2\eta_w^2C^2L^2}\varsigma^2.
\end{align}


Let $\eta_w\leq \min\left(1/L,\frac{NL^2}{2L^2+2},\frac{1-q}{3\sqrt{2}CLN}\right)$, we obtain $(1-q)^2-18N^2\eta_w^2C^2L^2\geq 18 C^2N^3L^2(2\eta_v\tau+\eta_w).$

Hence we have
\begin{align}
&\sum\limits_{t=0}^{K-1} \Bigg(\frac{\eta_v\tau}{4}\left\|\nabla_{\mathbf{v}}L(\bar{\mathbf{w}}(t),\{\mathbf{v}_i(t)\}_{i=1}^N)\right\|^2\nonumber\\
&\quad+\frac{\eta_w}{4}\left\|\nabla_{\mathbf{w}}L(\bar{\mathbf{w}}(t),\{\mathbf{v}_i(t+1)\}_{i=1}^N)\right\|^2\Bigg)\nonumber\\
&\quad+\frac{\eta_w}{4N}\sum_{t=0}^{K-1}\sum_{i=1}^N\mathbb{E}\|\bar{\mathbf{w}}(t)-\mathbf{w}_i(t)\|^2 \nonumber\\
\leq\quad&\sum\limits_{t=0}^{K-1} \Bigg(\frac{\eta_v\tau}{2}\left\|\nabla_{\mathbf{v}}L(\bar{\mathbf{w}}(t),\{\mathbf{v}_i(t)\}_{i=1}^N)\right\|^2\nonumber\\
&\quad+\frac{\eta_w}{4}\left\|\nabla_{\mathbf{w}}L(\bar{\mathbf{w}}(t),\{\mathbf{v}_i(t+1)\}_{i=1}^N)\right\|^2\Bigg)\nonumber\\
&\quad+\frac{\eta_w}{2N}\sum_{k=0}^{K-1}\sum_{i=1}^N\mathbb{E}\|\bar{\mathbf{w}}(t)-\mathbf{w}_i(t)\|^2 \nonumber\\
\leq\quad&\mathbb{E}[L(\bar{\mathbf{w}}(0),\{\mathbf{v}_i(0)\}_{i=1}^N)]-\mathbb{E} [L(\bar{\mathbf{w}}(t),\{\mathbf{v}_i(t)\}_{i=1}^N)]\nonumber\\
&\quad+\frac{K\eta_w^2L}{2N}\sigma_1^2+K\eta_v^3L^2(18\tau^3-15\tau^2-3\tau)\sigma_2^2\nonumber\displaybreak[0]\\
&\quad+\frac{K\eta_v^2\tau L}{2}\sigma_2^2+\frac{K\eta_w^2}{6N}\left(1+\frac{1}{L^2}\right)\sigma_2^2\nonumber\\
&\quad+\frac{K\eta_w^2}{2N}\left(1+\frac{1}{L^2}\right)\varsigma^2.
\end{align}
Dividing both sides by $\eta_w K/4$, we obtain
\begin{align}
&\frac{1}{K}\sum_{k=0}^{K-1}\mathbb{E}[M(t)]\nonumber\\
=\quad&\frac{1}{K}\sum\limits_{k=0}^{K-1} \frac{\eta_v\tau}{\eta_w}\mathbb{E}\left\|\nabla_{\mathbf{v}}L(\bar{\mathbf{w}}(t),\{\mathbf{v}_i(t)\}_{i=1}^N)\right\|^2\nonumber\\
&\quad+\mathbb{E}\left\|\nabla_{\mathbf{w}}L(\bar{\mathbf{w}}(t),\{\mathbf{v}_i(t+1)\}_{i=1}^N)\right\|^2\nonumber\allowdisplaybreaks\\
&\quad+\frac{1}{N}\sum_{i=1}^N\mathbb{E}\|\bar{\mathbf{w}}(t)-\mathbf{w}_i(t)\|^2\nonumber\\
\leq\quad& \frac{4\mathbb{E}[L(\bar{\mathbf{w}}(0),\{\mathbf{v}_i(0)\}_{i=1}^N)]-4\mathbb{E} [L(\bar{\mathbf{w}}(t),\{\mathbf{v}_i(t)\}_{i=1}^N)]}{K\eta_w}\nonumber\displaybreak[0]\\
&\quad+\frac{2\eta_w L}{N}\sigma_1^2+\frac{12\eta_v^3L^2}{\eta_w}(\tau-1)(6\tau^2-\tau)\sigma_1^2\nonumber\\
&\quad+\frac{2\eta_v^2\tau L}{\eta_w}\sigma_2^2+\frac{2\eta_w}{3 N}\left(1+\frac{1}{L^2}\right)\sigma_2^2\nonumber\\
&\quad+\frac{2\eta_w}{ N}\left(1+\frac{1}{L^2}\right)\varsigma^2\nonumber\allowdisplaybreaks\\
\leq\quad&\frac{4f\left(\bar{\mathbf{w}}(0),\{\mathbf{v}_i(0)\}_{i=1}^N\right)-4f\left({\mathbf{w}}^*,\{\mathbf{v}_i^*\}_{i=1}^N\right)}{K\eta_w}\nonumber\displaybreak[0]\\
&\quad+\frac{2\eta_w L}{N}\sigma_1^2+\frac{12\eta_v^3L^2}{\eta_w}(\tau-1)(6\tau^2-\tau)\sigma_1^2\nonumber\\
&\quad+\frac{2\eta_v^2\tau L}{\eta_w}\sigma_2^2+\frac{2\eta_w}{3 N}\left(1+\frac{1}{L^2}\right)\sigma_1^2\nonumber\\
&\quad+\frac{2\eta_w}{ N}\left(1+\frac{1}{L^2}\right)\varsigma^2.
\end{align}
This completes the proof of Theorem \ref{thm:loss_convergence}.

\section{Algorithm supplement}
This is a variation of Alg.~\ref{alg:PE-MA1}. 
\begin{algorithm}[h] 
\caption{PE-MA: Parallel update version} 
\label{alg:PE-MA2}
\begin{algorithmic}[1]
\STATE \textbf{Initialize:} $N$ agents; frozen pretrained model parameters $\pmb{\theta}_{\text{freeze}}$; initial shared adapter $w^0$; specific adapters $\{v_i^0\}_{i=1}^N$; local training epochs $\tau$; local datasets $\{\mathbb{D}_i\}_{i=1}^N$; mixing coefficient $\mu$; learning rates $\eta_w$, $\eta_v$
\STATE \textbf{Output:} $\{v_i^T\}_{i=1}^{N}, \{w_i^T\}_{i=1}^{N}$
\FOR{each round $t = 0, 1, \dots, K-1$}
    \STATE \textbf{// Local Training Phase}
    \FOR{each local epoch $q = 0, 1, \dots, \tau-1$}
        \FOR{each agent $i = 1, 2, \dots, N$}
            \FOR{each sample $(x, y) \in \mathbb{D}_i$}
                \STATE $\hat{y}_1 \leftarrow f_i(x; \pmb{\theta}_{\text{freeze}}, w_i(t,q))$ \hfill \textit{// via shared-adapter}
                \STATE $\hat{y}_2 \leftarrow f_i(x; \pmb{\theta}_{\text{freeze}}, v_i(t,q))$ \hfill \textit{// via personalized-adapter}
                \STATE $L_i(x,y) \leftarrow \ell_i\left(\mu \cdot \hat{y}_2 + (1-\mu) \cdot \hat{y}_1, y\right)$
            \ENDFOR
        \ENDFOR
        \STATE Update personalized adapter: $\mathbf{v}_{i}(t,q+1) \leftarrow \mathbf{v}_{i}(t, q) - \eta_v g_{\mathbf{v}}(\mathbf{w}_{i}(t,q), \mathbf{v}_{i}(t,q))$
        \STATE Update shared adapter: $\mathbf{w}_{i}(t,q+1) \leftarrow \mathbf{w}_{i}(t, q) - \eta_w g_{\mathbf{w}}(\mathbf{w}_{i}(t,q), \mathbf{v}_{i}(t,q))$
    \ENDFOR
    \STATE Set $\mathbf{v}_i(t+1,0)=\mathbf{v}_i(t,\tau)$
    \STATE \textbf{// Communication Phase}
    \FOR{each agent $i = 1, 2, \dots, N$}
        \STATE Broadcast $\mathbf{w}_{i}(t,\tau)$ to neighbors $\mathcal{N}_i$
        \STATE Aggregate shared adapter: $\mathbf{w}_{i}(t+1,0) \leftarrow  \sum_{j \in \mathcal{N}_i} \mathbf{w}_{j}(t,\tau)P_{ij}$
    \ENDFOR
\ENDFOR
\end{algorithmic}
\end{algorithm}

\end{document}